%% file: main.tex
\documentclass[a4paper,USenglish,cleveref, autoref, thm-restate]{lipics-v2021}
\graphicspath{ {./figs/} }
\usepackage[toc,page]{appendix}

\usepackage{tabu}
\usepackage{mathtools}
\usepackage{microtype}
\usepackage{color}
\usepackage[usenames,dvipsnames,svgnames,table]{xcolor}

\usepackage{mathrsfs}
\usepackage{amssymb,latexsym,amsmath}
\usepackage{amsthm, mathrsfs}
\usepackage{hyperref}
\usepackage{caption}
\usepackage{subcaption}

\usepackage{epsfig}
\usepackage{array,multirow}

\nolinenumbers
\newtheorem{problem}{Problem}

\include{pre}

\title{On Range Summary Queries}
\author{Peyman Afshani}{Aarhus University, Denmark }{peyman@cs.au.dk}{}{}
\author{Pingan Cheng}{Aarhus University, Denmark }{pingancheng.au.dk}{}{}
\author{Aniket Basu Roy}{Aarhus University, Denmark }{aniket.au.dk}{}{}
\author{Zhewei Wei}{Renmin University of China, China }{zhewei@ruc.edu.cn}{}{}
\authorrunning{P. Afshani, P. Cheng, A. B. Roy, Z. Wei} 

\Copyright{Peyman Afshani, Pingan Cheng, Aniket Basu Roy, Zhewei Wei} 



\category{} 

\begin{document}
\maketitle
\begin{abstract}
We study the query version of the approximate heavy hitter and quantile problems.
In the former problem, the input is a parameter $\varepsilon$ and a set $P$ of $n$ points in $\R^d$  
where each point is assigned a color from a set $C$, and the goal is to
build a structure such that given any geometric range $\gamma$,
we can efficiently find a list of approximate heavy hitters in $\gamma\cap P$,
i.e., colors that appear at least $\varepsilon |\gamma \cap P|$ times in $\gamma \cap P$,
as well as their frequencies with an additive error of $\varepsilon |\gamma \cap P|$.
In the latter problem, each point is assigned a weight from a totally ordered universe
and the query must output a sequence $S$ of $1+1/\varepsilon$ weights
such that the $i$-th weight in $S$ has approximate rank $i\varepsilon|\gamma\cap P|$, meaning,
rank $i\varepsilon|\gamma\cap P|$ up to an additive error of $\varepsilon|\gamma\cap P|$.
Previously, optimal results were only known in 1D~\cite{KZ.summary} but
a few sub-optimal methods were available in higher dimensions~\cite{AW.range.sampling,ach+12}. 

We study the problems for two important classes of geometric ranges: 3D halfspace and 3D dominance queries.
It is known that many other important queries can be reduced to these two, e.g., 
1D interval stabbing or interval containment, 
2D three-sided queries, 2D circular as well as 2D $k$-nearest neighbors queries.
We consider the real RAM model of computation where integer registers of size $w$ bits, 
$w = \Theta(\log n)$, are also available. 
For dominance queries, we show optimal solutions for both heavy hitter and
quantile problems: using linear space, we can answer both queries in time $O(\log n + 1/\varepsilon)$.
Note that  as the output size is $\frac{1}{\varepsilon}$, after investing the initial $O(\log n)$ searching time,
our structure takes on average $O(1)$ time to find a heavy hitter or a quantile!
For more general halfspace heavy hitter queries, the same optimal query time can be achieved by increasing the
space by an extra $\log_w\frac{1}{\varepsilon}$ (resp. $\log\log_w\frac{1}{\varepsilon}$) factor in 3D (resp. 2D).
By spending extra $\log^{O(1)}\frac{1}{\varepsilon}$ factors in both time and space,
we can also support quantile queries. 

We remark that it is hopeless to achieve a similar query bound for dimensions 4
or higher unless significant advances are made in the data structure side of
theory of geometric approximations.

\end{abstract}

\input{sections/sec1-intro}

\input{sections/sec2-prelim}

\input{sections/sec3-hh}

\input{sections/sec4-quantiles}

\input{sections/sec5-open}

\bibliography{ref}

\input{sections/appendices}

\end{document}

%% file: pre.tex
\newcommand{\ignore}[1]{} 
\newcommand{\para}[1]{\vspace{1mm}\noindent\textbf{#1}}

\newcommand{\bbig}{{\text{Big}}}

\newcommand{\sD}{\mathscr{D}}
\newcommand{\rD}{\mathcal{D}}

\newcommand{\C}{\mathscr{C}}
\newcommand{\sC}{\mathcal{C}}

\newcommand{\T}{\mathcal{T}}

\newcommand{\R}{\mathbb{R}}

\newcommand{\RR}{\mathcal{R}}

\newcommand{\cell}{\mathit{\Delta}}
\newcommand{\eps}{\mathcal{E}}
\newcommand{\sF}{\mathcal{F}} 
\newcommand{\sS}{\mathcal{S}} 

\DeclarePairedDelimiter\abs{\lvert}{\rvert}%

\def\T{\mathcal{T}}

%% file: sections/sec1-intro.tex
\section{Introduction}
\label{sec:sec1-intro}

Range searching is an old and fundamental area of computational geometry that
deals with storing an input set $P \subset \R^d$ of  $n$ (potentially weighted)
points in a data structure such that 
given a query range $\gamma$, one can answer certain questions about the subset of points inside $\gamma$.
Range searching is often introduced within a general framework that allows a very diverse set of questions to be answered.
For instance, if the points in $P$ have been assigned integer or real weights, then one can count the points in $\gamma$ (range counting),
sum the total weights of the points in $\gamma$ (weighted range counting), or find the maximum or minimum weight in $\gamma$
(range max or min queries). 

However, there are some important questions that cannot be answered 
within this general framework. 
Consider the following motivating example: our data includes the locations of houses in a
city as well as their estimated values and given a query range $\gamma$, we are interested in the
distribution of the house values within $\gamma$, for example, we might be interested to see if there's a
large inequality in house values or not. 
Through classical results, we can find the most expensive and the least
expensive houses (max and min queries), and the average value of the houses (by
dividing the weighted sum of the values by the total number of houses in
$\gamma$).  Unfortunately, this information does not tell us much about 
the distribution of the house values within $\gamma$, e.g., one cannot
compute the Gini index which is a widely-used measure of inequality
of the distribution.
Ideally, to know the exact distribution of values within $\gamma$, one must have all the
values inside $\gamma$, which in the literature is known as a 
\textit{range reporting} query which reports all the points inside the query range $\gamma$.
However, this could be an expensive operation, e.g., it can take $\Omega(n)$ time if the query 
contains a constant fraction of the input points. 
A reasonable alternative is to ask for a ``summary'' query, one that can  summarize the distribution.
In fact, the streaming literature is rich with many important notions of summary that are used to
concisely represent a large stream of data approximately but with high precision. 
Computing $\varepsilon$-quantiles can be considered as one of the most important concepts for a succinct approximation
of a distribution and it also generalizes many of the familiar concepts,
e.g., $0$-quantile, $0.5$-quantile, and $1$-quantile that are also known
as the minimum, the median, and the maximum of $S$.
We now give a formal definition below. 

\para{Quantile summaries.}
Given a sequence of values $w_1 \le  \cdots \le w_k$, a $\delta$-quantile, for $0 \le \delta \le 1$, is the value with 
rank $\lfloor \delta k\rfloor$.
By convention, $0$-quantile and $1$-quantiles are set to be the minimum and the maximum, i.e., $w_1$ and $w_k$ respectively.
An $\varepsilon$-quantile summary is then defined as the list of $1+\varepsilon^{-1}$ values where the
$i$-th value is the $ i\varepsilon$-quantile, for $i=0, \cdots, \varepsilon^{-1}$.
As we will review shortly, computing exact quantiles is often too expensive so instead we focus on 
approximations. 
We define an approximate $\varepsilon$-quantile summary (AQS) to be a sequence of 
$1+\varepsilon^{-1}$ values where the
$i$-th value is between the $(i-1)$-quantile and the $(i+1)$-quantile\footnote{
For $a\le 0$ (resp. $a\ge k$), we define the $a$-quantile to be the $0$-quantile (resp. $k$-quantile).}, for $i=0, \cdots, \varepsilon^{-1}$.
An approximate quantile summary with a reasonably small choice of $\varepsilon$ can give a
very good approximation of the distribution, e.g., see Figure~\ref{fig:ex}. 
It also has the benefit that the query needs to output only
$O(\varepsilon^{-1})$ values, regardless of the number of points inside the query range. 

To obtain a relatively precise approximation of the distribution, $\varepsilon$ needs to be chosen 
sufficiently small, and thus we consider it an additional 
parameter (and thus not a constant). 
This is also similar to the literature on streaming where the dependency on $\varepsilon$ is important. 

\begin{figure}[h]
  \newcommand{\cwidth}{0.32}
  \newcommand{\bwidth}{0.9}
  \begin{minipage}{\cwidth\textwidth}
      \begin{subfigure}[b]{\bwidth\textwidth}
          \includegraphics[scale=0.24]{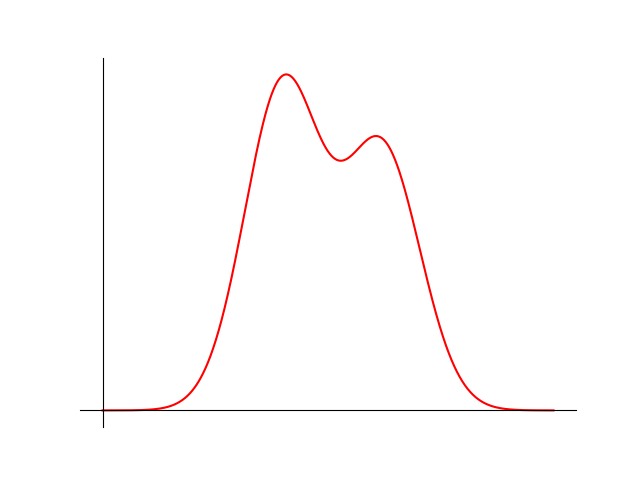}
          \caption{The distribution.}
      \end{subfigure}
  \end{minipage}
  \begin{minipage}{\cwidth\textwidth}
      \begin{subfigure}[b]{\bwidth\textwidth}
          \includegraphics[scale=0.24]{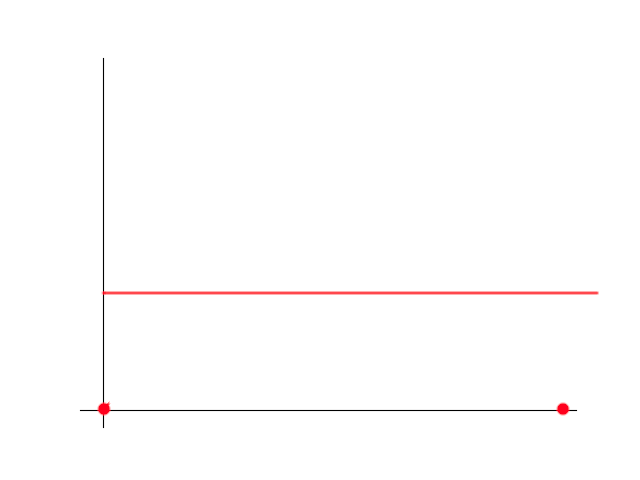}
          \caption{Classical results 
          (min., max. and average)}
      \end{subfigure}
  \end{minipage} 
  \begin{minipage}{\cwidth\textwidth}
      \begin{subfigure}[b]{\bwidth\textwidth}
          \includegraphics[scale=0.24]{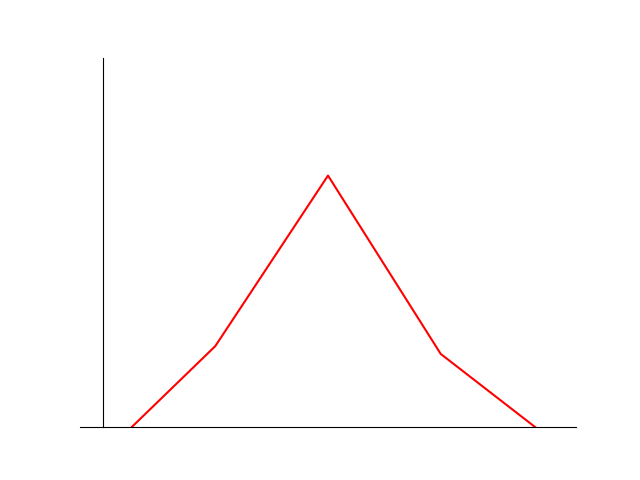}
          \caption{Approximate quintiles ($\varepsilon=0.2$).}
      \end{subfigure}
  \end{minipage} \\
  \begin{minipage}{\cwidth\textwidth}
      \begin{subfigure}[b]{\bwidth\textwidth}
          \includegraphics[scale=0.24]{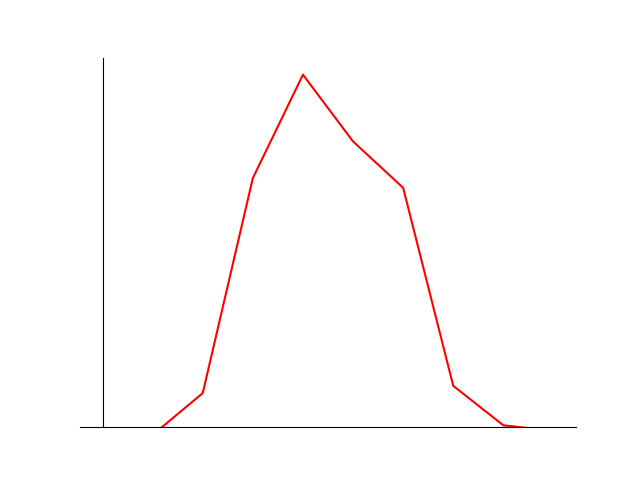}
          \caption{Approximate deciles ($\varepsilon=0.1$).}
      \end{subfigure}
  \end{minipage}
  \begin{minipage}{\cwidth\textwidth}
      \begin{subfigure}[b]{\bwidth\textwidth}
          \includegraphics[scale=0.24]{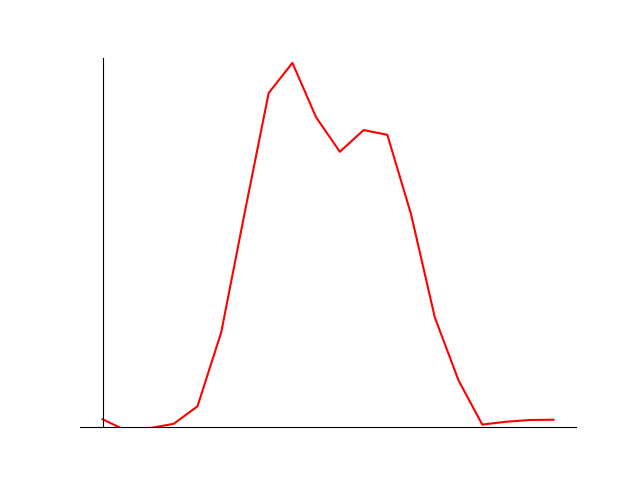}
          \caption{Approximate ventiles ($\varepsilon=0.05$).}
      \end{subfigure}
  \end{minipage}
  \begin{minipage}{\cwidth\textwidth}
      \begin{subfigure}[b]{\bwidth\textwidth}
          \includegraphics[scale=0.24]{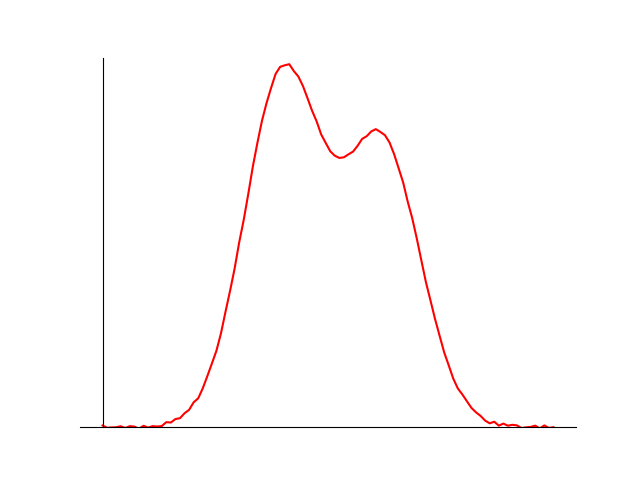}
          \caption{Approximate percentiles ($\varepsilon=0.01$).}
      \end{subfigure}
  \end{minipage}
  \caption{Comparing various approximations of a distribution.}
  \label{fig:ex}
\end{figure}

\subsection{Problem Definition, Previous Work, and Related Results}	
One of our main problems is the problem of answering approximate quantile summary (AQS) queries which is defined as follows. 

\begin{problem}[Approximate quantile summaries]
\label{pr:aqs}
Consider an input set $P$ of $n$ points in $\mathbb{R}^d$
where each point $p\in P$ is assigned a weight $w_p$ from a totally ordered universe.
Given a value $\varepsilon$, we are asked to build a structure such that given a query range
$\gamma$, it can return an AQS of $P\cap\gamma$ efficiently.
\end{problem}

It turns out that another type of ``range summary queries'' is extremely useful for building
data structures for AQS queries.

\para{Heavy hitter summaries.}
Consider a set $P$ of $k$ points where each point in $P$ is assigned a color from the set $[n]$. 
Let $f_i$ be the frequency of color $i$ in $P$, i.e., the number of times 
color $i$ appears among the points in $P$. 
A \textit{heavy hitter summary (HHS)} with parameter $\varepsilon$, is the list of all the colors $i$ with $f_i \ge \varepsilon k$ together
with the value $f_i$. 
As before, working with exact HHS will result in very inefficient data structures and thus once again we turn to approximations.
An \textit{approximate heavy hitter summary (AHHS)} with parameter $\varepsilon$ is a list, $L$, of colors
such that every color $i$ with $f_i \ge \varepsilon k$ is included in $L$ and furthermore, every color $i \in L$
is also accompanied with an approximation, $f'_i$, of its frequency such that 
 $f_i - \varepsilon k \le f'_i \le f_i + \varepsilon k$.

\begin{problem}[Approximate heavy hitters summaries]
\label{pr:ahhs}    
Consider an input set $P$ of $n$ points in $\mathbb{R}^d$ where 
each point in $P$ is assigned a color from the set $[n]$.
Given a parameter $\varepsilon$, we are asked to build
a structure such that  given a query $\gamma$, it can return 
an AHHS of the set $P \cap\gamma$. 
\end{problem}

Observe that in both problems, the output size of a query is $O(1/\varepsilon)$ in the worst-case.
Our main focus is to obtain data structures with the optimal worst-case query time of $O(\log n + \varepsilon^{-1})$. 
Note that it makes sense to define an \textit{output-sensitive} variant where the query time is 
$O(\log n +k)$ where $k$ is the output size. E.g., it could be the case for a AHHS query that the numbrer of heavy hitters
is much fewer than $\varepsilon^{-1}$. 
This makes less sense for AQS queries, since unless the distribution of weights inside the query range
$\gamma$ is almost constant, an AQS will have $\Omega(\varepsilon^{-1})$ distinct values. 
As our main focus is on AQS, we only consider AHHS data structures with the worst-case query time of $O(\log n + \varepsilon^{-1})$.

\para{A note about the notation.}
To reduce the clutter in the expressions of query time and space, we adopt the convention that 
$\log(\cdot)$ function is at least one, e.g., 
we define $\log_a b$ to be $\max\{1,\frac{\ln b}{\ln a}\}$ for any positive values $a,b$.

\subsubsection{Previous Results}
As discussed, classical range searching solutions focus on rather simple queries that can return
sum, weighted sum, minimum, maximum, or the full list of points contained in a given query range.
This is an extensively researched area with numerous results to cite and so we refer the reader 
to an excellent survey by Agarwal~\cite{Agarwal.survey16} that covers such classical results. 

However, classical range searching data structures cannot give 
detailed statistical information about the set of points contained inside the query region,
unless one opts to report the entire subset of points inside the query range, which could be very expensive if the set is large. 
Because of this, there have been a number of attempts to answer more informative queries. 
For example, ``range median'' queries have received quite a bit of attention~\cite{kms05,bgjs11,larsen:median}.
Note that the median is the same as $0.5$-quantile and 
thus these can be considered the first attempts at answering quantile queries.
However, optimal solution (linear space and logarithmic query time) to exact range median queries has only be found in 1D~\cite{bgjs11}.
For higher dimensions, to the best of our knowledge, 
the only known technique is to reduce the problem to several range counting instances~\cite{bgjs11,cz15mrs},
and it is a major open problem in the range searching field to find efficient data structures for exact range counting.
Due to this barrier, the approximate version of the problem~\cite{Jit.et.al.median} has been studied.

Data summary queries have also received some amount of attention, especially in the context of
geometric queries. 
Agarwal et al.~\cite{ach+12} showed that the heavy hitters summary (as well as a few other data summaries)
are ``mergeable'' and this gives a baseline solution for a lot of different queries in higher dimensions,
although a straightforward application of their techniques gives sub-optimal dependency on $\varepsilon$. 
In particular, for $d=2$ and for halfspace (or simplex) 
queries it yields a linear-space data structure with
$O(\frac{\sqrt{n}}{\varepsilon})$ query time. For $d=3$ the query time will be $O(n^{2/3}/\varepsilon)$. 
In general, in the naive implementation,
the query time will be $O(f(n)/\varepsilon)$ where $f(n)$ is the query time of the corresponding
``baseline'' range searching query (see Table~\ref{tab:results} for more information). 
A more efficient approach towards merging of summaries was taken by~\cite{doi:10.1137/16M1093604}
where they study the problem in a communication complexity setting, however, it seems possible to
adopt their approach to a data structure as well, in combination with standard application of 
partition trees; after building an optimal partition tree, for any node $v$ in the tree,
consider it as a player in the communication problem with 
the subset of points in the subtree of $v$  as its input. 
At the query time, after identifying $O(n^{2/3})$ subsets that cover the query range, 
the goal would be to merge all the summaries involved.
By plugging the results in~\cite{doi:10.1137/16M1093604}
this can result in a linear-space data structure with query time
of $\tilde{O}(n^{2/3} + n^{1/6}\varepsilon^{-3/2})$.

The issue of building optimal data structures for range summary queries was only tackled in
1D by Wei and Yi~\cite{KZ.summaryJ}. They built a data structure for answering a number of
summary queries, including heavy hitters queries, and showed it is possible to obtain an
optimal data structure with $O(n)$ space and $O(\log n + 1/\varepsilon)$ query time.
Beyond this, only sub-optimal solutions are available. 
Recently, there have been efforts to tackle ``range sampling queries'' where the goal is to extract
$k$ random samples from the set $|P \cap \gamma|$~\cite{AP.range.sampling, AW.range.sampling,hu2014independent}.
In fact, one of the main motivations to consider range sampling queries was to gain information about
the distribution of the point set inside the query~\cite{AP.range.sampling}. 
In particular, range sampling provides a general solution for obtaining a ``data summary'' and for example, 
it is possible to solve the heavy hitters query problem.
However, it has a number of issues, in particular,
it requires sampling at least $1/\varepsilon^2$ points from the set $|P \cap \gamma|$, 
and even then it will only provide a Monte Carlo type
approximation which means to boost the probabilistic guarantee, even more points need to be sampled. 
For example, to get a high probability guarantee, $\Omega(\varepsilon^{-2}{\log n})$ samples are required. 

\para{Type-2 Color Counting.}
These queries were introduced in 1995 by Gupta et al.~\cite{Coloured_reporting_counting_3sided} within
the area of ``colored range counting''. 
In this problem, given a set of colored points, we want to report the frequencies of all the colors
that appeared in a given query range.
This is a well-studied problem, but mostly in the orthogonal setting, see e.g.,~\cite{chan.colored.20}.

AHHS queries can be viewed as approximate type-2 color counting queries but with an additive error.
Consider a query with $k$ points. If we allow error $\varepsilon k$ in type-2 counting, then we can ignore colors with
frequencies fewer than $\varepsilon k$ but otherwise we have to report frequencies with error $\varepsilon k$,
which is equivalent to answering an AHHS query.

\para{Other Related Problems.}
Karpinski and Nekrich~\cite{KarpinskiN08} studied the problem of finding the most frequent colors
in a given (orthogonal) query range.
This problem has received further attention in the community~\cite{MR3126351,MR4264249,MR3000968}.
But the problem changes fundamentally when we introduce approximations.

\para{The Model of Computation.}
Our model of computation is the real RAM where we have access to real registers
that can perform the standard operations on real numbers in constant time, but
we also have access to $w = \Theta(\log n)$ bits long integer registers that can
perform the standard operations on integers and extra nonstandard operations
which can be implemented by table lookups
since we only need binary operations on fewer than $\frac{1}{2} \log n$ bits.
Note that our data structure works when the input coordinates are real numbers, however,
at some point, we will make use of the capabilities of our model of computation to 
manipulate the bits inside its integer registers. 

\subsection{Our Contributions}
Our main results and a comparison with the previously known results are shown in Table~\ref{tab:results}.

Overall, we obtain a series of new results for 3D AHHS and AQS query problems which
improve the current results via mergeability and independent range sampling~\cite{ach+12, AW.range.sampling} 
by up to a huge multiplicative $n^{\Omega(1)}$ factor in query time with almost the same linear-space usage.
This improvement is quite nontrivial and requires an innovative combination of known
techniques like the shallow cutting lemma, 
the partition theorem, $\varepsilon$-approximations,
as well as some new ideas like bit-packing for nonorthogonal queries,
solving AQS query problem using AHHS instances, 
rank-preserving geometric sampling and so on.

For dominance queries, we obtain the first optimal results.
When $\varepsilon^{-1}=O(\log n)$ our halfspace AHHS results are also optimal.
Note that for small values of $\varepsilon$, our halfspace AHHS results yield
significant improvements in the query time over the previous approaches. 
Along the way, we also show improved results of the above problems for 2D
as well as a slightly improved exact type-2 simplex color counting result.

{
\begin{table*}[h]
\centering
\setlength\extrarowheight{2.5pt}
\caption{Our main results
compared with Mergeability-based~\cite{ach+12} and 
Independent Range Sampling (IRS)-based~\cite{AW.range.sampling} solution.
The IRS-based solutions are randomized with success probability $1-\delta$
for a parameter $0<\delta<1$. 
$F$ is the number of colors of the input.
$w=\Theta(\log n)$ is the word size of the machine.
$\dagger$ indicates optimal solutions.}
\label{tab:results}
\tabulinesep=1.2mm
\begin{tabu}{ | m{2.5cm} | m{3.3cm}| m{4.2cm} | m{3.4cm} |} 
\hline
    \bf{Summary Query Types} \centering
& 
    \bf{Space} 
& 
    \bf{Query Time} 
& 
    \bf{Remark} \\
\hline
    \bf{Type-2 Simplex Color Counting} \centering
& 
    $O(n)$ 
&
    $O\left(n^{1-\frac{1}{d}}+\frac{n^{1-\frac{1}{d}}F^{\frac{1}{d}}}{w^{\alpha}}\right)$ 
& 
    New
\\
\hline
    \bf{3D AHHS \newline Halfspace} \centering
& 
    $O(n)$ \newline 
    $O(n)$ \newline 
    $O(n)$ \newline 
    $\boldsymbol{O(n\log_w \frac{1}{\varepsilon})}$ 
& 
    $O(\log n + \frac{1}{\varepsilon}n^{2/3})$ \newline 
    $\tilde{O}(n^{2/3} + \frac{1}{\varepsilon^{3/2}}n^{1/6})$ \newline 
    $O(\log n+\frac{1}{\varepsilon^2}\log\frac{1}{\delta})$ \newline 
    $\boldsymbol{O(\log n + \frac{1}{\varepsilon})}$ 
& 
    Mergeability-based~\cite{ach+12}  \newline 
    Monte Carlo~\cite{doi:10.1137/16M1093604}  \newline 
    IRS-based~\cite{AW.range.sampling} \newline 
    \textbf{New} 
\\
\hline
    \bf{3D AHHS \newline Dominance} \centering
& 
    $O(n)$ \newline 
    $O(n)$ \newline 
    $\boldsymbol{O(n)}$
& 
    $O(\log n + \frac{1}{\varepsilon}\log^3 n)$ \newline 
    $O(\log n + \frac{1}{\varepsilon^2}\log\frac{1}{\delta})$ \newline 
    $\boldsymbol{O(\log n + \frac{1}{\varepsilon})}$
& 
    Mergeability-based~\cite{ach+12}  \newline 
    IRS-based~\cite{AW.range.sampling} \newline 
    \textbf{New}$\dagger$
	\\
\hline
		\bf{3D AQS \newline Halfspace} \centering
	&
		$O(n)$ \newline 
		$O(n)$ \newline 
		$\boldsymbol{O(n\log^2\frac{1}{\varepsilon}\log_w\frac{1}{\varepsilon})}$
	& 
		$O(\log n + \frac{1}{\varepsilon}n^{2/3}\log (\varepsilon n))$ \newline 
		$O(\log n + \frac{1}{\varepsilon^2}\log\frac{1}{\delta})$ \newline 
		$\boldsymbol{O(\log n + \frac{1}{\varepsilon}\log^2\frac{1}{\varepsilon})}$
	& 
		Mergeability-based~\cite{ach+12}  \newline
		IRS-based~\cite{AW.range.sampling} \newline
		\textbf{New}
	\\
\hline
		\bf{3D AQS \newline Dominance} \centering
	&
		$O(n)$ \newline 
		$O(n)$ \newline 
		$\boldsymbol{O(n)}$
	& 
		$O(\log n + \frac{1}{\varepsilon}\log^3 n \log (\varepsilon n))$ \newline 
		$O(\log n + \frac{1}{\varepsilon^2}\log\frac{1}{\delta})$ \newline 
		$\boldsymbol{O(\log n + \frac{1}{\varepsilon})}$
	& 
		Mergeability-based~\cite{ach+12}  \newline
		IRS-based~\cite{AW.range.sampling} \newline
		\textbf{New}$\dagger$
	\\
\hline
\end{tabu}
\end{table*}

%% file: sections/sec2-prelim.tex
\section{Preliminaries}

In this section, we introduce the main tools we will use in our results.
For a comprehensive introduction to the tools we use, see Appendix~\ref{sec:premexact} and Appendix~\ref{sec:premsc}.

\subsection{Shallow Cuttings and Approximate Range Counting}
\label{sec:shallow}

Given a set $H$ of $n$ hyperplanes in $\R^3$, the level of a point $q\in\R^3$
is the number of hyperplanes in $H$ that pass below $q$.
We call the locus of all points of level at most $k$ the $(\le k)$-level
and the boundary of the locus is the $k$-level.
A shallow cutting $\C$ for the $(\le k)$-level of $H$ (or a $k$-shallow cutting for short)
is a collection of disjoint cells (tetrahedra) 
that together cover the $(\le k)$-level of $H$
with the property that every cell $C\in\C$ in the cutting intersects a set $H_C$,
called the conflict list of $C$
, of $O(k)$ hyperplanes in $H$.
The shallow cutting lemma is the following.

\begin{lemma}
\label{lem:sc}
	For any set of $n$ hyperplanes in $\R^3$ and a parameter $k$,
	there exists an $O(k/n)$-shallow cutting of size $O(n/k)$ that covers the $(\le k)$-level.
	The cells in the cutting are all vertical prisms unbounded from below
	(tetrahedra with a vertex at $(0,0,-\infty)$).
	
	Furthermore, we can construct these cuttings for all $k$ of form $a^i$ simultaneously in $O(n\log n)$ time
	for any $a>1$. Given any point $q\in\R^3$, we can find the smallest level $k$ that is above $q$
	as well the cell containing $q$ in $O(\log n)$ time.
\end{lemma}

The above can also be applied to dominance ranges, which are defined as below.
Given two points $p$ and $q$ in $\R^d$, $p$ dominates $q$ if and only if every
coordinate of $p$ is larger or equal to that of $q$.
The subset of $\R^d$ dominated by $p$ is known as a \textit{dominance range}. 
When the query range in a range searching problem is a dominance range, we refer
to it as a \textit{dominance query}.
As observed by Chan et al.~\cite{chan:revisit}, dominance queries can be simulated by a halfspace queries
and thus Lemma~\ref{lem:sc} applies to them. 
See Appendix~\ref{sec:premsc} for details.

We obtain the approximate version of the range counting result using shallow cuttings.

\begin{restatable}[Approximate Range Counting~\cite{AHZ.UB.CGTA}]{theorem}{approxrc}\label{thm:approxrc}
    Let $P$ be a set of $n$ points in $\R^3$. One can build a data structure of size
    $O(n)$ for halfspace or dominance ranges such that given a query range $\gamma$,
    one can report $|\gamma\cap P|$ in $O(\log n)$ time with error $\alpha|\gamma\cap P|$
    for any constant $\alpha>0$.
\end{restatable}

\subsection{$\varepsilon$-approximation}
Another tool we will use is $\varepsilon$-approximation,
which is a useful sampling technique:

\begin{definition}
	Let $(P,\Gamma)$ be a finite set system.
	Given any $0<\varepsilon<1$,
	a set $A\subseteq P$ is called an $\varepsilon$-approximation
	for $(P,\Gamma)$ if for any $\gamma\in\Gamma$,
	$
		\abs*{ \frac{|\gamma \cap A|}{|A|} - \frac{|\gamma \cap P|}{|P|} } \le \varepsilon.
	$
\end{definition}

The set $A$ above allows us to approximate the number of points of $\gamma \cap P$  with additive
error $\varepsilon |P|$ by computing $|\gamma \cap A|$ exactly; essentially, $\varepsilon$-approximations
reduce the approximate counting problem on the (big) set $P$ to
the  exact counting problem on the (small) set $A$.

It has been shown that small-sized $\varepsilon$-approximations for set systems formed by points and halfspaces/dominance ranges exist:
\begin{theorem}[$\varepsilon$-approximation~\cite{matousek2009geometric, p16coresets}]\label{thm:eapp}
	There exist $\varepsilon$-approximations of size $O(\varepsilon^{-\frac{2d}{d+1}})$ and $O(\varepsilon^{-1}\log^{d+1/2}\varepsilon^{-1})$ 
	for halfspace and dominance ranges respectively.
\end{theorem}

%% file: sections/sec3-hh.tex
\section{Approximate Heavy Hitter Summary Queries}
We solve approximate quantile summary (AQS) queries
using improved results for approximate heavy hitter summary (AHHS) queries.
We sketch the main ideas of our new AHHS solutions in this section
and refer the readers to Appendix~\ref{sec:hh} for details.
For the clarity of description,
we use $\varepsilon_0$ to denote the target error for the AHHS queries.
We will reserve $\varepsilon$ as a general error parameter.
We show the following.

\begin{restatable}{theorem}{thmddd}\label{thm:3d}
    For $d=3$, the approximate halfspace heavy hitter summary queries can be answered using $O(n\log_w(1/\varepsilon_0))$ space
    and with the optimal $O(\log n + 1/\varepsilon_0)$ query time. 
\end{restatable}

\begin{restatable}{theorem}{thmdd}\label{thm:2d}
    For $d=2$, the approximate halfspace heavy hitter summary queries can be answered using $O(n\log\log_w(1/\varepsilon_0))$ space
    and with the optimal $O(\log n + 1/\varepsilon_0)$ query time. 
\end{restatable}

\begin{restatable}{theorem}{thmdomhh}\label{thm:domhh}
    For $d=2,3$, the approximate dominance heavy hitter summary queries can be answered using the optimal $O(n)$ space
    and with the optimal $O(\log n + 1/\varepsilon_0)$ query time. 
\end{restatable}

\subsection{Base Solution}
The above results are built from a \textit{base solution}, which 
solves the following problem:
\begin{restatable}{problem}{probcoarse}[Coarse-Grained AHHS Queries]\label{pr:coarse}
  Let $P$ be a set of points in $\R^d$, 
  each associated with a color. 
  The problem is to store $P$ in a structure such that
  given a query range $q$, one can estimate the frequencies of colors in $q\cap P$
  with an additive error up to $\varepsilon |P|$ efficiently for some parameter $0<\varepsilon<1$.
\end{restatable}

Note that here we allow more error (since the error is defined in the entire point set).
To solve Problem~\ref{pr:coarse},
one crucial component we need is a better (exact) type-2 color
counting structure for halfspaces.
We combine several known techniques in a novel way with bit-packing
to get the following theorem. See Appendix~\ref{sec:type-2} for details.

\begin{restatable}{theorem}{partition}\label{thm:partition}
    Given an integer parameter $F$, a set $P$ of $n$ points in $\R^d$ where 
    each point is assigned a color from the set $[F]$, 
	one can build a linear-sized data structure, such that given a query simplex $q$,
	it can output the number of times each color appears in $P \cap q$ in total time  
	$\max\{O(n^{(d-1)/d}), O(n^{(d-1)/d} F^{1/d}/w^\alpha)\}$, for some appropriate
	constant $\alpha$ and word size $w$.
\end{restatable}

The main idea for getting a base solution is relatively straightforward.
We group colors according to their frequencies where
each group contains colors of roughly equal frequencies.
However, we have to be careful about the execution and the analysis is a bit tricky. 
For example, if we place all the points in one copy of the data structure of Theorem~\ref{thm:partition},
then we will get a sub-optimal result.
However, by grouping the points correctly, and being stringent about the analysis, we can obtain
the following.
\begin{restatable}{theorem}{thmbase}\label{thm:base}
  For $d\ge 3$,
  Problem~\ref{pr:coarse} for simplex queries (the intersection of $d+1$ halfspaces)
  can be solved with $O(X)$ space for $X=\min\{|P|, \varepsilon^{-\frac{2d}{d+1}}\}$
  and a query time of 
\[
	O\left( \frac{|P|^{1-\frac{2}{d-1}}}{w^{\alpha}\varepsilon^{\frac{2}{d-1}}} \right)+ O\left(X^{\frac{d-1}{d}}  \right)
\]
  where $w$ is the word-size of the machine and $\alpha$ is some positive constant.
\end{restatable}

The main challenge is that we have two cases for the size of an $\varepsilon$-approximation on $n$ 
points since it is bounded by
$\min\left\{n,O(\varepsilon^{-\frac{2d}{d+1}}) \right\}$
and also two cases for the query time of Theorem~\ref{thm:partition}. 
However, the main idea is that since the total error budget is $\varepsilon|P|$,
we can afford to pick a larger error parameter $\varepsilon_i=\frac{\varepsilon|P|}{|P_i|}$,
where $P_i$ is the set of points with color $i$. 
The details are presented in Appendix~\ref{sec:hh}. 

\subsection{Solving AHHS Queries}
We first transform the problem into the dual space. So the point set $P$ becomes a set $H$ of hyperplanes
and any query halfspace becomes a point $q$.
We want to find approximate heavy hitters of hyperplanes of $H$ below $q$.
Here, we remark that obtaining a data structure with $O(n \log {\frac{1}{\varepsilon_0}})$ space is not too difficult:
build a hierarchy of shallow cuttings covering level $2^i/\varepsilon_0$ for $i=0,1,\cdots,\log (\varepsilon_0n)$
of the arrangement of $H$.
For each shallow cutting cell $\cell$, we build the previous base structure for the conflict list $\sS_\cell$
for a parameter $\varepsilon=\varepsilon_0/c$ for a big enough constant $c$.
Then, observe that for queries below level $\varepsilon_0^{-1}$, 
we can spend $O(\log n + \frac{1}{\varepsilon_0})$ time to find all the
hyperplanes passing below the query and answer the AHHS queries explicitly and 
also for shallow cutting levels above level $\varepsilon_0^{-3/2}$,
the total amount of space used by the base solution is $O(n)$. 
Thus, it turns out that the main difficulty lies in handling the levels between $\varepsilon_0^{-1}$ and 
$\varepsilon_0^{-3/2}$.

To reduce the space to $O(\log_w\frac{1}{\varepsilon_0})$, 
recall that in the query time of the base structure, we have two terms $O(1/(\varepsilon_0 w^{\alpha}))$ and $O(X^{2/3})$. 
Observe that we can afford to set $\varepsilon$ to be roughly $\varepsilon_0/w^{\alpha}$ and the first term will still be 
$O(\varepsilon_0^{-1})$
because we are at level below $\varepsilon_0^{-3/2}$, we have $X<\varepsilon_0^{-3/2}$ 
and so the second term will always be $O(\varepsilon_0^{-1})$!
The effect of setting $\varepsilon=\varepsilon_0/w^{\alpha}$ is that now the base structure we built for a cell
can output frequencies with a factor of $w^{\alpha}$ more precision, meaning it can be used for a factor of $w^{\alpha}$ many more levels.
So  we only need to build the base structure for shallow cuttings built for a factor of $w^{\alpha}$!
This gives us the $O(n\log_w\frac{1}{\varepsilon_0})$ space bound.
Of course, here the output has size $O(\varepsilon^{-1})=O(w^{\alpha}\varepsilon_0^{-1})$ and we cannot afford to examine all these colors.
The final ingredient here is that we can maintain a list of $O(\varepsilon_0^{-1})$ candidate colors using shallow cuttings built for a factor of $2$.

We remark that although the tools are standard,
the combination of the tools and the analysis are quite nontrivial.
Also when we have $\Theta(1/\varepsilon_0)$ heavy hitters, our query time is optimal.
It is an interesting open problem if the query time can be made output sensitive.

%% file: sections/sec4-quantiles.tex
\section{Approximate Quantile Summary Queries}
\label{sec:quantiles}

In this section, we solve Problem~\ref{pr:aqs}.
We first show a general technique that uses our solution to AHHS queries solution to obtain an efficient
solution for AQS queries. 
We show that for 3D halfspace and dominance ranges we can convert the solution for AHHS queries to a solution 
for AQS queries with an $O(\log^2{\frac{1}{\varepsilon})}$ blow up in space and time.
Then in Section~\ref{sec:domq}, we present an optimal solution for dominance ranges
based on a different idea.

First, we show how to solve AQS queries using the AHHS query solution.
We describe the data structure for halfspaces, since as we have mentioned before, the same can be applied to dominance
ranges in 3D as well.
The high level idea of our structure is as follows:
We first transform the problem into the dual space.
This yields the problem instance where we have $n$ weighted hyperplanes and given a query point $q$,
we would like to extract an approximate quantile summary for the hyperplanes that pass below $q$. 
To do this, we build hierarchical shallow cuttings.
For each cell in each cutting,
we collect the hyperplanes in its conflict list
and then divide them into $O(\frac{1}{\varepsilon_0})$ groups
according to the increasing order of their weights.
Given a query point in the dual space,
we first find the cutting and the cell containing it,
and then find an approximated rank of each group, within the subset below the query. 
This is done by generating an AHHS problem
instance and applying Theorem~\ref{thm:3d}.
We construct the instance in a way such that
the rank approximated will only have error small enough
such that we can afford to scan through the groups
and pick an arbitrary hyperplane in corresponding groups
to form an approximate $\varepsilon_0$-quantile summary.

\subsection{The Data Structure and the Query Algorithm}
We dualize the set $P$ of $n$ input points
which gives us a set $H=\overline{P}$ of $n$ hyperplanes.
We then build a hierarchy of shallow cuttings where
the $i$-th shallow cutting, $\C_i$, is a $k_i$-shallow cutting where $k_i= \frac{2^i}{\varepsilon_0}$, for $i=0,1,2,\cdots,\log (\varepsilon_0n)$.
Consider a cell $\cell$ in the $i$-th shallow cutting and its
conflict list $\sS_{\cell}$.
Let $\epsilon=\frac{\varepsilon_0}{c}$ for a big enough constant $c$.
We partition $\sS_\cell$ into 
$t=\frac{1}{\epsilon}$ groups
$G_1,G_2,\cdots,G_t$ sorted by weight, meaning, the weight
of any hyperplane in $G_j$ is no larger than 
that of any hyperplane in $G_{j+1}$ for $j=1,2,\cdots,t-1$.

For each group $G_j$,
we store the smallest weight among the hyperplanes it contains, as its representative.
To make the description shorter,
we make the simplifying assumption that
$t$ is a power of $2$ (if not, we can add some dummy groups). 
We arrange the groups $G_j$ as the leaves of a 
balanced binary tree $\mathcal{T}$ and
let $V(\T)$ be the set of vertices of $\T$.
Next, we build the following set $A_\Delta$ of colored hyperplanes, associated with $\Delta$:
Let $\varepsilon' = \frac{\epsilon}{\log^2t}$.
For every vertex $v \in V(\T)$, let $G_v$ to be the set of all the hyperplanes contained in the subtree of
$v$; we add an $\varepsilon'$-approximation, $E_v$, of $G_v$ to $A_\Delta$ with color $v$.
Using Theorem~\ref{thm:3d}, we store the points dual to hyperplanes in $A_\Delta$ in a data structure $\Psi_\Delta$ for
AHHS queries with error parameter $\varepsilon'$.
This completes the description of our data structure.

\para{The query algorithm.}
A given query $q$ is answered as follows.
Let us quickly go over the standard parts:
We consider the query in the dual space and thus $q$ is considered to be a point. 
Let $k$ be the number of hyperplanes passing below $q$. 
Observe that by Theorem~\ref{thm:approxrc}, we can find a $(1+\alpha)$ factor approximation, $k^*$,
of $k$ in $O(\log n)$ time for any constant $\alpha$, using a data structure that consumes linear space. 
This allows us to find the first $k_i$-shallow cutting $\C_i$ with $k_{i-1}< k\le k_i$.
The cell $\cell\in\C_i$ containing $q$ can also be found in $O(\log n)$ time using a standard
point location data structure (e.g., see~\cite{ac09}).

The interesting part of the query is how to handle the query after
finding the cell $\cell$.
Let $H_q$ be the subset of $H$ that lies below $q$.
Recall that $\sS_\cell$ is the subset of $H$ that intersects $\cell$.
The important property of $\cell$ is that $H_q \subset \sS_\cell$
and also $|\sS_\cell| = O(|H_q|) = O(k)$.

We query the data structure $\Psi_\cell$ built for $\cell$
to obtain a list of colors and their approximate counts where the additive error
in the approximation is at most $\varepsilon' |A_\cell|$.
To continue with the description of the query algorithm, let us use the notation
$g_j$ to denote the subset of $G_j$ that lies below $q$,
and let $g = \cup_{j=1}^t g_j$ and thus $|g| = k$.

Note that while the query algorithm does not have direct access to $g$, or $k$,
we claim that using the output of the data structure $\Psi_\cell$, we can calculate the
approximate rank of the elements of $g_i$ within $g$ up to an additive error 
of $\varepsilon_0 k$.
Again, we can use tree $\mathcal{T}$ to visualize this process.
Recall that in $\Psi_\cell$, every vertex $v \in V(\T)$ represents a unique color 
in the data structure $\Psi_\cell$ and the data structure returns an AHHS summary with error parameter $\varepsilon'$.
This allows us to estimate the number of elements of $E_v$ that pass below $q$ with error $\varepsilon' |A_\cell|$
and since $E_v$ is an $\varepsilon'$-approximation of $G_v$, this 
allows us to estimate the number of elements of $G_v$ that pass below $q$ with error at most $2\varepsilon'|A_\cell|$.
Consider the leaf node that represents $g_j \subset G_j$ and the 
path $\pi$ that connects it to the root of $\T$. 
The approximate rank, $r_j$, of $g_j$ is calculated as follows.
Consider a subtree with root $u$ that hang to the left of the path $\pi$
(as shown in Figure~\ref{fig:app-rank}). 
If color $u$ does not appear in the output of the AHHS query, then we can conclude that
at most $2\varepsilon' |A_\cell|$ of its hyperplanes pass below $q$ and in this case we do nothing. 
If it does appear in the output of the AHHS query, then we know the number of hyperplanes
in its subtree that pass below $q$ up to an additive error of $2\varepsilon' |A_\cell|$ and in this
case, we add this estimate to $r_j$.
In both cases, we are off by an additive error of $2\varepsilon' |A_\cell|$.
We repeat this for every subtree that hangs to the left of $\pi$.
The number of such subtree is at most $\log t$ and thus the total error is at most
$2\varepsilon' |A_\cell| \log t$.
Now observe that
\[
    2\varepsilon' |A_\cell| \log t = 2\frac{\epsilon}{\log^2t} \cdot \log t |\sS_\cell| \cdot \log t = O(\varepsilon k) = O\left( \frac{\varepsilon_0 k}{c} \right) \le \varepsilon_0 k 
\]
which follows by setting $c$ large enough and observing the fact that
$|A_\cell| \le \log t |\sS_\cell|$ since every
hyperplane in $\sS_\cell$ is duplicated $\log t $ times. 

\begin{figure}[h]
	\centering
	\includegraphics[width=0.3\textwidth]{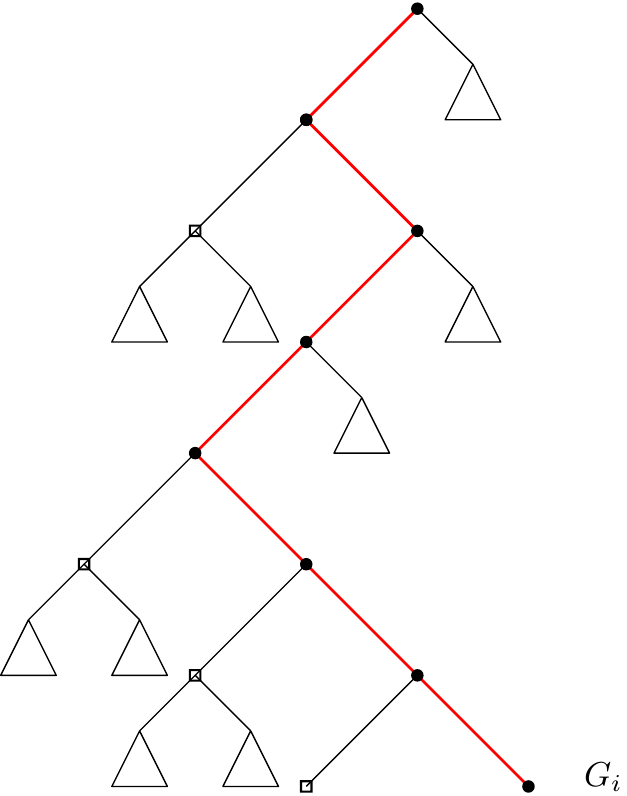}
	\caption{Compute the Approximate Rank of a Group:
	The approximated rank of $G_i$ is calculated
	as the sum of all the approximate counts of square nodes.}
	\label{fig:app-rank}
\end{figure}

We are now almost done. 
We just proved that
in each $g_i$, we know the rank of its elements within $g$ up to an additive
error of $\varepsilon_0 k$.
This means that picking one element from each $G_i$ gives us a super-set of an
AQS; in the last stage of the query algorithm we simply prune the unnecessary
elements as follows:
We scan all the leave in $\mathcal{T}$
from left to right, i.e., consider the group $G_j$ 
for $j=1$ to $t$ and compute the quantile summary in a straightforward fashion.
To be specific, we initialize a variable $j'=0$ and then 
consider $G_j$, for $j=1$ to $t$.
The first time $r_j$ exceeds a quantile boundary, i.e., $r_j \ge j'\varepsilon_0 k^*$,
we add the hyperplane with the lowest weight in $G_j$ 
to the approximate $\varepsilon_0$-quantile summary, 
and then increment $j'$. 

\subsubsection{Analysis}
Based on the previous paragraph, the correctness is established.
Thus, it remains to analyze the space and query complexities. 
We start with the former. 

\para{Space Usage.}
Consider the structure $\Psi_\cell$ built for cell $\cell$ from  a $k_i$-shallow cutting $\C_i$.
Observe that $\sum_{v \in V(\T)}|G_v| = |\sS_\cell| \log t$ since in the sum every hyperplane will
be counted $\log t$ times.
$E_v$ is an $\varepsilon'$-approximation of $G_v$  and thus
\begin{align}
    |E_v| \le \min\left\{ \varepsilon'^{-3/2},G_v \right\}\label{eq:ev}
\end{align}
which implies
\begin{align}
    |A_\cell| = \sum_{v \in V(\T)} |E_v| \le \min\left\{ \varepsilon'^{-3/2}2t,|\sS_\cell|\log t \right\}\label{eq:acell}
\end{align}
where the first part follows as there are at most $2t$ vertices in $\T$ and the second part follows from (\ref{eq:ev}).
We build an instance of Theorem~\ref{thm:3d} on the set $A_\cell$ which by 
Theorem~\ref{thm:3d} uses $O(|A_\cell|\log_w\frac{1}{\varepsilon'})$ space. 
Assuming $\cell$ belongs to a $k_i$-shallow cutting $\C_i$,
we have $|\sS_\cell| = O(k_i)$ and  
there are $O(n/k_i)$ cells in $\C_i$.
Observe that
\begin{align}
    \sum_{\cell \in \C_i} |A_\cell| =& \sum_{\cell \in \C_i} \min\left\{ \varepsilon'^{-3/2}2t,|\sS_\cell|\log t \right\}= \sum_{\cell \in \C_i} O\left( \min\left\{ \varepsilon'^{-3/2}t,k_i \log t \right\} \right) = \nonumber\\
    &O\left( \min\left\{ \frac{n}{k_i}\varepsilon_0^{-3},n \log\frac{1}{\varepsilon_0} \right\} \right).\label{eq:ci}
\end{align}
Thus, the total space used for $\C_i$ is
\[
    O\left( \min\left\{ \frac{n}{k_i}\varepsilon_0^{-3}\log_w\frac{1}{\varepsilon_0},n \log_w\frac{1}{\varepsilon_0}\log\frac{1}{\varepsilon_0} \right\} \right).\label{eq:ci}
\]

Finally, observe that there can be at most $O(\log\frac{1}{\varepsilon_0})$ levels where the second term dominates;
to be specific, at least when  $k_i$ exceeds $\varepsilon_0^{-4}$, the first term dominates and the total space used by those
levels is $O(n)$ as $k_i$'s form a geometric series.
So the total space usage of our structure is $O(n\log^2\frac{1}{\varepsilon_0}\log_w\frac{1}{\varepsilon_0})$.

\para{Query Time.}
By Lemma~\ref{lem:sc},
we can find the desired cutting cell in time $O(\log n)$.
Next, we query the data structure $\Psi_\cell$ which by 
Theorem~\ref{thm:3d} uses $O(\log n + \varepsilon'^{-1}) = O(\log n 
+ \frac{\epsilon}{\log^2t}) = O(\log n + \frac{1}{\varepsilon_0}\log^2\frac{1}{\varepsilon_0})$ query time. 
Scanning the groups and pruning the output of the data structure $\Psi_\cell$ takes asymptotically smaller
time and thus it can be absorbed in the above expression. 
Therefore, we obtain the following result.

\begin{restatable}{theorem}{thmhqddd}\label{thm:3dhq}
    Given an input consisting of an error parameter $\varepsilon_0$, and a set $P$ of
    $n$ points in $\R^3$ where each point $p\in P$ is associated with a weight
    $w_p$ from a totally ordered universe,
    one can build a data structure that uses
    $O(n\log^2\frac{1}{\varepsilon_0}\log_w\frac{1}{\varepsilon_0})$ space such that
    given any query halfspace $h$,
	it can answer an AQS query with parameter $\varepsilon_0$ in 
	time $O(\log n + \frac{1}{\varepsilon_0}\log^2\frac{1}{\varepsilon_0})$.
\end{restatable}

For the case of 2D, we can just replace $\Psi_\cell$ with the structure in Theorem~\ref{thm:2d},
and we immediately get the following:

\begin{restatable}{theorem}{thmhqdd}\label{thm:2dhq}
    Given an input consisting of an error parameter $\varepsilon_0$, and a set $P$ of
    $n$ points in $\R^2$ where each point $p\in P$ is associated with a weight
    $w_p$ from a totally ordered universe,
    one can build a data structure that uses
    $O(n\log^2\frac{1}{\varepsilon_0}\log\log_w\frac{1}{\varepsilon_0})$ space such that
    given any query halfspace $h$,
	it can answer an AQS query with parameter $\varepsilon_0$ in 
	time $O(\log n + \frac{1}{\varepsilon_0}\log^2\frac{1}{\varepsilon_0})$.
\end{restatable}

\subsection{Dominance Approximate Quantile Summary Queries}
\label{sec:domq}
Now we turn our attention to dominance ranges.
We will show a structure similar to that for halfspace queries.
The main difference is that we now
use exact type-2 color counting as an auxiliary structure to estimate the rank of each group.
This saves us roughly $\log^2\frac{1}{\varepsilon_0}$ factors for both space and query time
and so we can answer quantile queries in the optimal $O(\log n + \frac{1}{\varepsilon_0})$ time.
To reduce the space to linear, we need more ideas.
We first present a suboptimal but simpler structure
to demonstrate our main idea.
Then we modify this structure to get the desired optimal structure.
We use shallow cuttings in the primal space.

\subsubsection{A Suboptimal $O(n\log\log\frac{1}{\varepsilon_0})$ Space Solution}
We first describe a data structure that solves the dominance AQS problem
with $O(n\log\log\frac{1}{\varepsilon_0})$ space and the optimal $O(\log n + \frac{1}{\varepsilon_0})$ query time.

\para{Rank-Preserving Approximation for Weighted Points.}
Let $S$ be a weighted point set where every point has been assigned a
weight from a totally ordered universe. 
Let $r_S(p)$ be the rank of a point $p$ in the set $S$.
Consider a geometric set system $(P, \sD)$,
where $P$ is a set of weighted points in $\R^3$ and $\sD$ is a family of subsets of $P$
induced by 3D dominance ranges.
We mention a way to construct a sample $A$ for $P$ and a parameter $\epsilon$ such that
\begin{align}
    \left|\frac{r_{P\cap\rD}(p)}{|P|} - \frac{r_{A\cap\rD}(p)}{|A|}\right|\le \epsilon \label{eq:neps}
\end{align}
for any point $p\in P$ and any range $\rD\in\sD$.
First note that taking an $\epsilon$-approximation for $P$ does \textbf{not} work
since it does not take the weights of $P$ into consideration.
Our simple but important observation is that we can lift the points $P$ into 4D by
adding their corresponding weights as the fourth coordinate.
Let us call this new point set $P'$ and let $(P', \sD')$ be the set system in 4D induced by 4D dominance ranges.
Consider an $\epsilon$-approximation $A'$ for $P'$ and let $A$ be the projection of $A'$ into the
first three dimensions (i.e., by removing the weights again). 
$A$ will be our sample for $P$ and to distinguish it from an unweighted approximation, we call
it \textit{rank-preserving $\epsilon$-approximation}. 
Indeed, for any point $p\in P$ with weight $w_p$ and any $\rD\in\sD$,
$r_{P\cap\rD}(p)$ (resp. $r_{A\cap\rD}(p)$) is equal to the number of points in $P'$ (resp. $A'$) contained in 4D dominance range $\rD\times(-\infty,w_p)$.
By the definition of $\epsilon$-approximation, property (\ref{eq:neps}) holds.

We now turn our attention to the AQS for 3D dominance queries. 

\para{The Data Structure and The Query Algorithm.}
Similar to the structure we presented for halfspace queries,
we  build $\frac{2^i}{\varepsilon_0}$-shallow cuttings
for $i=0,1,\cdots,\log (\varepsilon_0 n)$.
Let $\kappa=O(1)$ be the constant
such that $O(\frac{1}{\varepsilon_0}\log^{\kappa}\frac{1}{\varepsilon_0})$ is 
the size of the $\varepsilon_0$-approximation for dominance ranges in 4D.
Consider one $k$-shallow cutting $\C$.
We consider two cases:
\begin{itemize}
    \item If $k\le\frac{1}{\varepsilon_0}\log^{\kappa}\frac{1}{\varepsilon_0}$,
for each cell $\cell$ in the cutting $\C$,
we collect the points in its conflict list $\sS_\cell$
and divide them into $t=\frac{1}{\epsilon}$ groups
$G_1,G_2,\cdots,G_t$ according to their weights (meaning, the weights
in $G_i$ are no larger than weights in group $G_{i+1}$)
where $\epsilon=\frac{\varepsilon_0}{c}$ for a big enough constant $c$
as we did for halfspace queries.

\item For $k>\frac{1}{\varepsilon_0}\log^{\kappa}\frac{1}{\varepsilon_0}$,
    we take a rank-preserving $\epsilon$-approximation of $\sS_\cell$ first,
    and then divide the approximation into $t=\frac{1}{\epsilon}$ groups, just like the above case.
Again, for each group, we store the smallest weight among the points it contains.
\end{itemize}
We build the following structure for each cell $\cell$.

Let $N$ be the number of points in all the $t$ groups we generated for a cell $\cell$.
We collect groups $G_{i\cdot\alpha+1}, G_{i\cdot\alpha+2},
\cdots, G_{(i+1)\cdot\alpha}$
into a cluster $\sC_i$ for each $i=0, 1, \cdots, t/\alpha-1$
where $\alpha = (\log\log\frac{1}{\varepsilon_0})^3$.
For each group $j$ in cluster $\sC_i$ for $j=1,2,\cdots, \alpha$,
we color the points in the group with color $j$.
Then we build the following type-2 color counting structure $\Psi_i$ for $\sC_i$.
Let $N_i$ be the total number of points in $\sC_i$:
\begin{itemize}
    \item First, we store three predecessor search data structures, one for each
        coordinate. This allows us to map the input coordinates as well as the query coordinates
        to rank space. 

    \item Next, we build a grid of size $\sqrt[3]{N_i}\times\sqrt[3]{N_i}\times\sqrt[3]{N_i}$
    such that each slice contains $\sqrt[3]{N_i^2}$ points.
    For each grid point, we store the points it dominates in a frequency vector using the compact representation.
\item Finally, we recurse on each grid slab (i.e., three recursions, one for each dimension). 
    The recursion stops when the number of points 
    in the subproblem becomes smaller than $N_* = N_i^{\eta}$ for some small enough constant $\eta$.

\item 
    For these ``leaf'' subproblems,
    note that the total number of different answers to queries is bounded by $O(N_i^{3\eta})$.
    We build a lookup table which records the corresponding frequency vectors for these answers.
    Note that since at every step we do a rank space reduction, the look up can be simply done in 
    $O(1)$ time, after reducing the coordinates of the query to rank space. 
\end{itemize}

\para{The query algorithm.}
Given a query $q$, we first locate the grid cell $C$ containing $q$
and this gives us three ranks.
Using the ranks for $x$ and $y$, we obtain an entry and using the rank of $z$,
we find the corresponding word and the corresponding frequency vector stored in the lower corner of $C$. 
We get three more frequency vectors by recursing to three subproblems.
We merge the three frequency vectors to generate the final answer.
This completes the description of the structure we build for each family $\sC_i$.

To answer a query $q$, we first find the first shallow cutting level above $q$
and the corresponding cutting cell $\cell$.
We then query the data structure described above to get the count the number of
points dominated by $q$ in each of the $t$ groups.
Then by maintaining a running counter,
we scan through the $t$ groups from left to right
to construct the approximate $\varepsilon_0$-quantile summary.

\para{Space Usage.}
For the space usage, note that there are $N_i$ grid points in each recursive level
and the recursive depth is $O(1)$.
There are $\alpha$ colors and the frequency of a color is no more than $N_i$.
So the total number of words needed to store frequency vectors is $O(N_i \frac{\alpha\log N_i}{w})$.
When the problem size is below $N_i^{\eta}$,
for each subproblem, we store a lookup table using $O(N_i^{3\eta}\frac{\alpha\log N_i}{w})$ words.
So the total number of words used for the bottom level is 
$O(\frac{N_i}{N_i^{\eta}})\cdot O(N_i^{3\eta}\frac{\alpha\log N_i}{w}) 
= O(N_i\frac{N_i^{2\eta}\alpha\log N_i}{w})$.
Note that by our construction and $\varepsilon_0\ge\frac{1}{n}$, $N=O(\frac{1}{\varepsilon_0}\log^{\kappa}\frac{1}{\varepsilon_0})$,
$\alpha=(\log\log\frac{1}{\varepsilon_0})^3\le(\log\log n)^3$
and $N_i=O(\alpha\frac{N}{1/\varepsilon_0})
=O(\alpha\log^{\kappa}\frac{1}{\varepsilon_0})
=O(\alpha\log^{\kappa}n)$.
Since by assumption, $w=\Omega(\log n)$, by picking $\eta$ in $N_*=N_i^{\eta}$ to be a small enough constant,
the space usage for frequency vectors satisfy
\[
f(N_i) =
\begin{cases}
	3\sqrt[3]{N_i}f(\sqrt[3]{N_i^2}) + O(\frac{N_i}{w^{1-o(1)}}),\textrm{for } N_i \ge N_* \\
	O(\frac{N_i}{w^{1-\beta}}), \textrm{otherwise}
\end{cases},
\]
for some constant $0<\beta<1$,
which solves to $O(\frac{N_i(\log N_i)^3}{w^{1-\beta}})=O(\frac{N_i}{w^{1-\tau}})$ for some constant $0<\tau<1$.
Since the recursive depth is $O(1)$, the space usage for all the predecessor searching structures is $O(N_i)$.
Therefore the space usage of $\Psi_i$ is $O(N)$.
So the total space for each shallow cutting cell $\cell$ is bounded by $\frac{N}{N_i}\cdot O(N_i)=O(N)$.

For $k_i\ge \frac{1}{\varepsilon_0}\log^{\kappa}\frac{1}{\varepsilon_0}$,
$N=O(\frac{1}{\varepsilon_0}\log^{\kappa}\frac{1}{\varepsilon_0})$.
So the total space usage for them is bounded by 
\[
	\sum_{i=\kappa\log\log\frac{1}{\varepsilon_0}}^{\varepsilon_0 n} O\left(\frac{n}{k_i}\right) \cdot O(N)
	= \sum_{i=\kappa\log\log\frac{1}{\varepsilon_0}}^{\varepsilon_0 n} O\left(\frac{n\varepsilon_0}{2^i}\right)\cdot 
	O\left(\frac{1}{\varepsilon_0}\log^{\kappa}\frac{1}{\varepsilon_0}\right) = O(n).
\]
For $k_i < \frac{1}{\varepsilon_0}\log^{\kappa}\frac{1}{\varepsilon_0}$,
$N=k_i$ and so we have space bound 
\[
	\sum_{i=0}^{\kappa\log\log\frac{1}{\varepsilon_0}}O\left(\frac{n}{k_i}\right)\cdot O(N)
	=\sum_{i=0}^{\kappa\log\log\frac{1}{\varepsilon_0}}O\left(\frac{n}{k_i}\right)\cdot O(k_i)
	=O\left(n\log\log\frac{1}{\varepsilon_0}\right).
\]
This completes our space bound proof.

\para{Query Time.}
For the query time, we first spend $O(\log n)$ time to find an appropriate shallow cutting level and the corresponding cell
by the property of shallow cuttings.
Then we query $\Psi_i$ for $i=0,1,\cdots,t/\alpha-1$ to estimate the count for each group in the cell.
For each $\Psi_i$, note that each predecessor searching takes $O(\log N_i)$ time.
Also each frequency vector can fit in one word and so we can merge two frequency vectors in time $O(1)$.
This gives us the following recurrence relation for the query time
\[
g(N_i) =
\begin{cases}
	3g(\sqrt[3]{N_i^2}) + O(\log N_i),\textrm{for } N_i \ge N_* \\
	O(\log N_i), \textrm{otherwise}
\end{cases},
\]
which solves to $O((\log N_i)^3)=O((\log\log\frac{1}{\varepsilon_0})^3)=O(\alpha)$.
Since we need to query $t/\alpha$ such data structures
to get the count for all groups, the total query time for count estimation is $O(t)=O(1/\varepsilon_0)$.
Then we scan through the groups and report the approximate quantiles which takes again $O(1/\varepsilon_0)$ time.
So the total query time is $O(\log n + \frac{1}{\varepsilon_0})$.

\para{Correctness.}
Given a query $q$,
let $k$ be the actual number of points dominated by $q$.
By the property of shallow cuttings, 
we find a cell $\cell$ containing $q$ in the shallow cutting level $k_i$ above it
such that $k\le k_i\le 2k$.
When $k_i < \frac{1}{\varepsilon_0}\log^{\kappa}\frac{1}{\varepsilon_0}$,
after we estimate the count in each group,
since the estimation is exact and each group has size $\frac{|\sS_\cell|}{t}=\frac{\varepsilon_0|\sS_\cell|}{c}$,
each quantile we output will have error at most $\frac{\varepsilon_0|\sS_\cell|}{c}$.
For $k_i \ge \frac{1}{\varepsilon_0}\log^{\kappa}\frac{1}{\varepsilon_0}$,
we introduce error $\epsilon|\sS_\cell|$ in the $\epsilon$-approximation,
but since we use exact counting for each group,
the total error will not increase as we add up ranks of groups.
So the total error is at most $\frac{2\varepsilon_0|\sS_\cell|}{c}$.
In both cases, the total error is at most $\varepsilon_0 k$ for a big enough $c$.

\subsubsection{An Optimal Solution for 3D Dominance AQS}
In this section, we modify the data structure in the previous section to reduce the space usage to linear.
It can be seen from the space analysis
that the bottleneck is shallow cuttings with $k_i\le\frac{1}{\varepsilon_0}\log^{\kappa}\frac{1}{\varepsilon_0}$.
For the structures built for these levels,
the predecessor searching structures take linear space at each level
which leads to a super linear space usage in total.
To address this issue,
we do a rank space reduction for points in the cells of these levels
before constructing $\Psi_i$'s
so that we can use the integer register to spend sublinear space
for the predecessor searching structures.

\para{Rank Space Reduction Structure.}
We consider the cells in the $\frac{1}{\varepsilon_0}\log^{\kappa}\frac{1}{\varepsilon_0}$-shallow cutting.
Let $A=\log^{\kappa+1}\frac{1}{\varepsilon_0}$.
For the points in the conflict list $\sS_\cell$ of a shallow cutting cell $\cell\in\C$,
we build a grid of size $A\times A\times A$
such that each slice of the grid contains $O(1/(\varepsilon_0\log\frac{1}{\varepsilon_0}))$ points.
The coordinate of each grid point consists of the ranks of its three coordinates
in the corresponding dimensions.
For each of the $O(\frac{A}{\varepsilon_0\log(1/\varepsilon_0)})$ points in $\sS_\cell$, we round it down to the closest grid point
dominated by it.
This reduces the coordinates of the points down to $O(\log \log\frac{1}{\varepsilon_0})$ bits
and now we can apply the sub-optimal solution from the previous subsection 
which leads to an $O(n)$ space solution. 
To be more specific, 
we build the hierarchical shallow cuttings for 
$k_i\le \frac{1}{\varepsilon_0}\log^{\kappa}\frac{1}{\varepsilon_0}$
locally for hyperplanes in $\sS_\cell$ and apply the previous solution with 
a value $\varepsilon'=\varepsilon/c$ for a large enough constant $c$. 

\para{Query Algorithm and the query time.}
The query algorithm is similar to that for the previous suboptimal solution.
The only difference is that when the query $q$ is in a shallow cutting level smaller 
than $\frac{1}{\varepsilon_0}\log^{\kappa}\frac{1}{\varepsilon_0}$,
we use the rank space reduction structure to reduce $q$ to the rank space.
Let $q'$ be the grid point obtained after reducing $q$ to rank space.
Observe that the set of points dominated by $q$ can be written as the union of 
the points dominated by $q'$ and the subset of points dominated by $q$ in three
grid slabs of $A$ that contain $q$.
We get an $\varepsilon'$-quantile for the former set using the data structure implemented on the
grid points.
The crucial observation is that there are $O(\varepsilon_0^{-1}/\log \frac{1}{\varepsilon_0})$ points
in the slabs containing $q$ and thus we can afford to build an approximate $\varepsilon'$-quantile summary of these
points in $O(\frac{1}{\varepsilon_0})$ time. 
We can then merge these two quantiles and return the answer as the result.
By setting $c$ in the definition of $\varepsilon'$ small enough, we make sure that
the result is a valid $\varepsilon_0$ quantile summary. 
This also yields a query time (after locating the correct cell $\cell$ in the shallow cutting)
of $O(\frac{1}{\varepsilon_0})$.

\para{Correctness.}
Since we build shallow cutting $k_i\le\frac{1}{\varepsilon_0}\log^{\kappa}\frac{1}{\varepsilon_0}$
inside each cell in $\frac{1}{\varepsilon_0}\log^{\kappa}\frac{1}{\varepsilon_0}$-shallow cutting,
the transformed coordinates are consistent.
As we described above, this introduces error to the counts $\Psi_i$'s outputs,
but since we correct the error explicitly afterwards, the counts we get are still exact.
The remaining is the same as the suboptimal solution and so our structure finds $\varepsilon_0$-quantile properly.

\para{Space Usage.}
For the rank space reduction structure,
we need to store a predecessor searching structure for the query,
which takes space linear in the number of slices which is $O(A)$.
We build this structure for each cell in the $A/\varepsilon_0$-shallow cutting level
and there are $O(n\varepsilon_0/A)$ cells in total,
and so the space usage is $O(n\varepsilon_0)$.
Building shallow cuttings inside each cell will only
increase the space by a constant factor by the property of shallow cuttings.

For each $\Psi_i$, by our analysis in the suboptimal solution,
the frequency vectors will take $O(\frac{N}{w^{1-\tau}})$ space.
Now since the coordinates of the points and queries are integers of size at most $A$,
it takes $O(\log A)=O(\log\log\frac{1}{\varepsilon_0})$ bits to encode a coordinate.
Since the word size is $w=\Omega(\log n)$,
we need only $O(\frac{N_i\log A}{w})$ space to build the predecessor searching structures for $\Psi_i$.
In total, we spend $O(\frac{N}{w^{1-o(1)}})$ space for each shallow cutting level less than $A/\varepsilon_0$.
So, the total space usage is $O(n)$.
We conclusion this section by the following theorem.

\begin{restatable}{theorem}{thmdomq}\label{thm:domq}
    Given an input consisting of a parameter $\varepsilon_0 > 0$, and a set $P$ of
    $n$ points in $\R^3$ where each point $p\in P$ is associated with a weight
    $w_p$ from a totally ordered universe,
    one can build a data structure that uses the optimal $O(n)$ space
	such that given any dominance query $\gamma$,
	the data structure can answer an AQS query with parameter $\varepsilon_0$ in
     the optimal query time of $O(\log n + \frac{1}{\varepsilon_0})$.
\end{restatable}

%% file: sections/sec5-open.tex
\section{Open Problems} 
Our results bring many interesting open problems.
First, for type-2 color counting problems,
we showed a linear-sized structure for simplex queries.
It is not clear if the query time can be reduced with more space.
It is an intriguing open problem to figure out the correct space-time tradeoff for the problem.
Note that our query time in Theorem~\ref{thm:partition} depends on the number of colors in total.
It is unclear if the query time can be made output-sensitive.
This seems difficult and unfortunately there seems to be no suitable lower bound techniques to settle the problem.
Furthermore, since improving exact simplex range counting results is already very challenging,
it makes sense to consider the approximate version of the problem with multiplicative errors.

Second, for heavy-hitter queries, there are two open problems.
In our solution, the space usage is optimal with up to some extra polylogarithmic factor (in $\frac{1}{\varepsilon}$).
An interesting challenging open problem is if the space usage can be made linear.
On the other hand, our query time is not output-sensitive. 
Technically speaking,
there can be less than $1/\varepsilon$ heavy hitters, 
and in this case, it would be interesting to see if $O(\log n + k)$ query time can be obtained
for $k$ output heavy hitters with (close to) linear space\footnote{We thank an anonymous referee for suggesting 
the ``output-sensitive'' version.}.

Third, for AQS queries,
our data structure for halfspace ranges is suboptimal.
The main reason is that we need a type-2 range counting solution as a subroutine.
For halfspace ranges, our exact type-2 solution is too costly,
and so we have to switch to an approximate version.
This introduces some  error
and as a result, we need to use a smaller error parameter,
which leads to extra polylogarithmic factors in both time and space.
In comparison, we obtain an optimal solution for dominance AQS queries
through exact type-2 counting.
Currently, it seems quite challenging to improve the exact type-2 result for halfspace queries
and some different ideas probably are needed to improve our results.

Finally, it is also interesting to investigate approximate quantile summaries, or heavy hitter summaries (or other data summaries
or data sketches used in the streaming literature) for a broader category of geometric ranges. 
In this paper, our focus has been on very fast data structures, preferably those with optimal $O(\log n + \frac{1}{\varepsilon})$ query time,
but we know such data structures do not exist for many important geometric ranges.
For example, with linear space, simplex queries require $O(n^{(d-1)/d})$ time and there are some matching lower bounds. 
Nonetheless, it is an interesting open question whether approximate quantile or heavy hitter summary can be built for
simplex queries in time $O(n^{(d-1)/d} + \frac{1}{\varepsilon})$  using linear or near-linear space;
as we review in the introduction, the general approaches result in sub-optimal query times of 
$O(n^{(d-1)/d} \cdot \frac{1}{\varepsilon})$  or 
$O(n^{(d-1)/d} + \frac{1}{\varepsilon^2})$.

%% file: sections/appendices.tex
\begin{appendices}

\section{The Partition Theorem, Cuttings, and Exact Range Counting}
\label{sec:premexact}

The partition theorem is a standard tool for simplex range searching.
It is originally proved by Matou\v sek~\cite{Matousek.efficient.par} in early 1990s
in the following form.
\begin{restatable}[Matou\v{s}ek's Partition Theorem \cite{Matousek.efficient.par}]{theorem}{thmmat}\label{thm:Mat}
	Let $P$ be a point set of size 
    	$n$ in $\R^d$ ($d\ge 2$),  and 
	$s$, $2\le s \le n$,  be an integer parameter and set $t = \frac{n}{2s}$. 
    	One can partition $P$ into $t$ subsets $P_1, \cdots, P_t$ such that
    	for every $1 \le i \le s$, $s \le |P_i| \le 2s$, and $P_i$ is enclosed in a
    	simplex $\Delta_i$ such that any hyperplane in $\R^d$ has 
    	crossing number $O(t^{1-1/d})$, meaning, it intersects $O(t^{1-1/d})$ of the simplices. 
\end{restatable}

By directly applying this theorem, we can solve simplex range searching problems
in near optimal time (up to $\log^{O(1)}n$ factors).
The main reason is that the constant hidden in the crossing number
will explode if we build a partition tree of super constant levels
using this theorem.
To get the optimal query time,
Chan~\cite{Chan.ParTree} refined and improved the partition theorem
and proved the following version.
\begin{restatable}[Chan's Partition Refinement Theorem~\cite{Chan.ParTree}]{theorem}{thmchan}\label{thm:Chan}
	Let $P$ be a point set of size 
    	$n$ in $\R^d$ ($d\ge 2$),  and 
	$H$ be a set of $m$ hyperplanes in $\R^d$. 
	Suppose we are given $\tau$ disjoint cells covering $P$, 
	such that each cell contains at most $2n/\tau$ points of $P$ and each hyperplane in $H$ crosses at most $\ell$ cells. 
	Then for any $b$, we can subdivide every cell into $O(b)$ disjoint subcells, for a total of at most $b\tau$ subscells, 
	such that each subcell contains at most $2n/(b\tau)$ points of $P$, 
	and each hyperplane in $H$ crosses at most the following total number of cells: 
	\[
      		O((b\tau)^{1-1/d} + b^{1-1/(d-1)}\ell + b\log \tau \log m).
   	 \]
\end{restatable}

Note that the crossing number in Chan's version has some extra terms ($b^{1-1/(d-1)}l$ and $b\log\tau\log m$).
These terms are dominated by the first term ($(b\tau)^{1-1/d}$) when $\tau=\log^{\Omega(1)} n$
but otherwise, the crossing number bound in Chan's bound can be worse than Matou\v sek's.
For most applications of the partition theorem,
we will build a partition tree of level $\Omega(\log\log n)$ and in that case,
these extra terms can be safely ignored.
Specifically, for the simplex range counting problem,
by applying Chan's method, 
we obtain the following optimal result for linear space structures.
\begin{restatable}[Small-Space Exact Simplex Range Counting~\cite{Matousek.RS.hierarchal.93}]{theorem}{simplexrc}\label{thm:simplexrc}
    Let $P$ be a set of $n$ points in $\R^d$. One can build a data structure of size 
    $O(n)$ for simplex ranges such that given a query range $\gamma$, one can report 
    $|\gamma\cap P|$ in $O(n^{(d-1)/d})$ time.
\end{restatable}

Given two points $p,q\in\R^d$, we say $p$ dominates $q$, denoted $q\ge p$, if all coordinates of $p$ is
no smaller than those of $q$.
A dominance range $\gamma$ is specified by a point $p_{\gamma}$
and is defined to be $\{q\in\R^d:q\le p_\gamma\}$.
When the context is clear, we use $p_\gamma$ and $\gamma$ interchangeably.
We will be interested in range counting for halfspace and dominance ranges in this paper.
But since these two types of ranges are special cases of simplex ranges
and so Theorem~\ref{thm:simplexrc} applies.

If we want to answer queries faster in polylogarithmic time,
one way is to spend more space.
This is done by using point-line duality~\cite{deberg2008computational} 
and applying the technique of cuttings.

Given a simplex $\cell\subset\R^d$ and a set $H$ of hyperplanes,
we denote the set of hyperplanes of $H$ intersecting the interior of $\cell$ by $\sS_\cell$,
which is called the conflict list of $\cell$.
An $(1/r)$-cutting for $H$ is defined to be a collection $\C$ of pairwise interior disjoint $d$-dimensional
closed simplicies that together cover $\R^d$ such that $|\sS_\cell|\le n/r$ for all $\cell\in\C$.
The size of a cutting $\C$, denoted by $|\C|$, is defined to be the number of simplicies it has.
The optimal size bound and construction algorithm of cuttings were developed after
a series of work by the pioneers in the computational geometry community~\cite{ClaDCG87,hw87,mat91cutting,agarwal91geometric,Chazelle.Friedman}
and culminated in the work of Chazelle~\cite{Chazelle.cutting}.
\begin{restatable}[Cutting Lemma~\cite{Chazelle.cutting}]{lemma}{cutting}\label{lem:cutting}
	Given a set $H$ of $n$ hyperplanes in $\R^d$, for any $r\ge 1$, there exists an $(1/r)$-cutting for $H$
	of size $O(r^d)$. The cutting as well as the conflict lists can be constructed in deterministically $O(nr^{d-1})$ time.
\end{restatable}

Using Theorem~\cite{Chazelle.cutting}, we can obtain the following fast query time result for simplex range counting.
\begin{restatable}[Fast-Query Exact Simplex Range Counting~\cite{Matousek.RS.hierarchal.93}]{theorem}{fastrc}\label{thm:fastrc}
    	Let $P$ be a set of $n$ points in 2D (resp. 3D). One can build a data structure of size 
   	 $O(n^d)$ for simplex ranges such that given a query range $\gamma$, one can report 
    	$|\gamma\cap P|$ in $O(\log n)$ time.
\end{restatable}

\section{Shallow Cuttings and Approximate Range Counting}
\label{sec:premsc}
To define shallow cuttings more formally.
We borrow the definition of shallow cuttings for general algebraic surfaces in~\cite{AHZ.UB.CGTA}.
Let $\sF$ be a collection of continuous and totally defined algebraic functions $f:\mathbb{R}^d\rightarrow\mathbb{R}$
for constant dimension $d$ and degree $\Delta$.
Each function $f\in \sF$ defines a continuous surface in $\mathbb{R}^{d+1}$.
Given any point $p=(p_1,p_2,\cdots,p_{d+1})\in\R^{d+1}$,
we say that $f$ passes below $p$ if $f(p_1,p_2,\cdots,p_d)\le p_{d+1}$.

The collection of all surfaces in $\sF$ partitions $\R^{d+1}$ into a subdivision consisting of
disjoint cells (faces of dimensions $0,1,\cdots,d+1$) that together cover the entire $\R^{d+1}$.
We call this subdivision the arrangement of $\sF$.
Given any point $p\in\R^{d+1}$, its level is defined by the number of functions $f\in\sF$
that passes below it.
We define the $(\le k)$-level of $\sF$ to be the closure of all points in $\R^{d+1}$
with level at most $k$.

A shallow cutting $\C$ for the $(\le k)$-level of $\sF$ (or a $k$-shallow cutting for short)
is a collection of disjoint cells that together cover the $(\le k)$-level of $\sF$
with the property that every cell $C\in\C$ in the cutting intersects a set $\sF_C$ of $O(k)$ functions in $\sF$.
We call $\sF_C$ the conflict list of $C$.

For a set $\sF$, we say it is shallow cuttable if it has the following properties:
\begin{enumerate}
	\item For each $k\in\mathbb{N}$, there exists a $(\le k)$-shallow cutting for $\sF$
	and the number of cells the shallow cutting has is bounded by $O(|\sF|/k)$.
	\item The shallow cutting can be constructed in time $O(|\sF|\log |\sF|)$.
	\item Given a hierarchy of $\alpha^i$-shallow cuttings, there exists a linear sized structure
	such that given any query point $q$, we can find the $\alpha^i$-shallow cutting $\C$ 
	and the cell $C\in\C$ that contains $q$ with the smallest $i$ in time $O(\log|\sF|)$.
\end{enumerate}

For halfspace and dominance range searching problems,
shallow cuttings are often used after translating the problem to the dual space,
in which the roles of inputs and outputs ``switch''.
It has been shown that both halfspace and dominance ranges are shallow cuttable in 2D and 3D~\cite{AHZ.UB.CGTA}.
More generally, we can count the number of points approximately in linear space and logarithmic time
for shallow cuttable ranges.
\approxrc*

We mention that the concept of shallow cuttings can also be defined in the primal space,
i.e., for a set $P$ of points.
For example, for dominance ranges,
given any point $p\in P$, we can define its level to be the number of points in $P$ dominated by $p$.
A (primal) $k$-shallow cutting for $P$ is a set $\C$ of points in $\R^d$
satisfying that each point in $P$ with level at most $k$ is dominated by a point in $\C$
and each point in $\C$ dominates $O(k)$ points in $P$. 
It has been shown that such shallow cuttings can be constructed efficiently
and we can find a shallow cutting level and a cell in a shallow cutting hierarchy efficiently
as in the shallow cutting in the dual space~\cite{at18}.

\section{Type-2 Colored Simplex Range Counting}
\label{sec:type-2}
In this section, we study the ``type-2'' colored range counting problem for simplex queries in $\R^d$.
In the problem, the input is a set of $n$ points in $\R^d$, and each point is assigned one color from a set of $F$ colors. 
We want to build a data structure such that given a query simplex $q$, 
we can report the number of points appearing in $q$ for each color.
In other words, we want to output a ``frequency vector'' for the colors.
Formally, we define frequency vectors as follows.

\begin{definition}[Frequency Vector]
	Given a point set $P$ with each point being assigned one of $F$ colors,
	the frequency vector of $P$ is defined to be $\mu(P) = (\mu_1,\dots,\mu_F)$ 
	where $\mu_i$ is the frequency of the $i$-th color in $P$.
\end{definition}

Before giving the full details, we describe 
the overall idea of our data structure.
Depending on different values of $F/w$, we will build different structures.
When $F/w > \frac{n}{\log^4n}$,
we build a partition tree $T$ using Matou\v sek's Theorem~\ref{thm:Mat}.
Otherwise,
we build a partition tree $T$ using Chan's Theorem~\ref{thm:Chan}.
For every vertex $v$ of $T$, 
we denote the partitioning simplex corresponding to $v$ by $v(\RR)$.
We store the colored points at the leaves of $T$ explicitly. 
Whereas, for every internal vertex $v$, we only store the frequency vector of the points in $v(\RR)$. 
The frequency vectors are represented in a compact way that we explain later. 

The query algorithm is as follows.
Given a query $q$, 
we query the corresponding partition tree $T$.
During the process, we maintain a running frequency vector, 
i.e., the total of all the different colors seen so far (in a compact representation). 
We start from the root $r$ of $T$. 
If $q$ fully contains $v(\RR)$ then we add the frequency vector of $v(\RR)$ to the running frequency vector. 
If $q$ does not contain $v(\RR)$ but has a non-empty intersection with it then we traverse its children. 
If $q$ and $v(\RR)$ have empty intersection we do nothing.
In case $v$ is a leaf of $T$, we explicitly count the frequency of every color in $v(\RR)$
and update the frequency vector. 

The total query time amounts to accessing the \emph{compact representations} 
of the frequency vectors for the internal vertices and the \emph{explicit counting} for the leaves which are traversed.

We now present the details. We start with the notion of compact representation. 

\subparagraph{Compact representations.}
Consider a frequency vector $\mu(\mu_1, \dots, \mu_F)$ and let $s=\sum_{i=1}^F \mu_i$.
We would like to store $\mu$ using $O(F\log(2(s+2F)/F))$ bits
while supporting the following two operations:
one, an ``add operation'' to add the representations of two such vectors to obtain a representation of the addition,
and two, an ``extract'' operation to retrieve the $i$-th frequency, $\mu_i$, for a given $i$ in constant time. 

We first describe the encoding using a third character \# besides 0 and 1:
We encode the $\mu_i$'s in binary and place the character \# between them. 
The representation can be
easily made binary by simply re-encoding the characters 0, 1, and \# in binary which increase the size of the
encoding by a factor of 2. 
Let $L=2\sum_{i=1}^F (\log (\mu_i+2) + 1)$, i.e., the length of the representation in bits. 
By the inequality of arithmetic and geometric means and $\sum_{i=1}^F\mu_i=s$,
\[
	L=2\sum_{i=1}^F \log (\mu_i +2)+ 1 = 2F + 2\log\prod_{i=1}^F(\mu_i+2) \le 2F+2F\log\frac{s+2F}{F}=2F\log\frac{2(s+2F)}{F}.
\]
We pack the representation in $O(\lceil L/w\rceil )$ words
and for each packed word, 
we store the number of the $\mu_i$'s it contains.
Let $c_j$ be the number of frequencies encoded in the $j$-th word.
Observe that $c_j = O(w)$ and thus, we can store them in a data structure for prefix sums
with constant query time. 
This concludes the description of the packed representation. 

Note that two representations can be added in a straightforward way using non-standard word operations
or using tabulations\footnote{E.g., pack into words of size  $\frac{\log n}{2}$
    and then build a table of size $n$ which supports the add operation in constant time.
}.

Now consider the extract operation with index $i$. 
We need to find the smallest index $j$ such that the prefix sum up to $j$ exceeds $i$.
However, as each $c_j$ is at most $w$, it can be encoded in unary and the extraction problem becomes equivalent to 
the \texttt{select} operation in bit vectors which can be performed in $O(1)$ time
after building an index that takes linear time~\cite{rrr07}.
To summarize, we have shown the following.

\begin{lemma}
\label{lem:compact}
	For any set of at most $s$ points
	with each point being assigned one of $F$ colors,
	we can encode its frequency vector with
	a compact representation of $O(1+\frac{F}{w}\log\frac{2(s+2F)}{F})$ words,
	where $w$ is the word size.
	Furthermore, we can add two compact representations in time $O(1+\frac{F}{w}\log\frac{2(s+2F)}{F})$
	and extract the frequency of the $i$-th color in time $O(1)$ for $1\le i\le F$.
\end{lemma}

\subparagraph{The Data Structure.}
Now we describe the details of our data structure.
As mentioned before, we build partition trees
based on an earlier version of Matou\v{s}ek (Theorem~\ref{thm:Mat}) 
as well as a more recent (refined and simplier)
version of Chan (Theorem~\ref{thm:Chan}). 
The main reason why we need both structures is although Chan's version
is superior in many aspects,
it is not able to handle the case when the partition tree is very shallow
(which can happen in our application).
Fortunately, Matou\v sek's early version is able to give the optimal bound in this case.
So we will use different partition trees for different cases.

\partition*

\begin{proof}
We build two partition trees based on the value of $\frac{F}{w}$.
The base case for both cases is the same:
When the size of the subproblem reaches below $F/w$ we stop
and examine all points by brute force.

When $\frac{F}{w}\ge \frac{n}{\log^4n}$,
we build a partition tree $T$ of Matou\v sek's version (Theorem~\ref{thm:Mat})
with fanout $t=w^{\delta}$ for some small enough $\delta >0$.
For each node in $T$, we store a frequency vector 
of the points rooted at that node in the compact representation.
The query time $Q(n)$ satisfies the following recurrence relation
\[
  Q(n) = O(t^{\frac{d-1}{d}}) Q\left(\frac{n}{t}\right) + O(t)\cdot O\left(1 + \frac{F}{w}\log\frac{2(n/t+2F)}{F}\right).
\]

After running for $k$ steps we have, 
\begin{align}
  Q(n) \le c_0^k t^{k\frac{d-1}{d}}Q\left(\frac{n}{t^k}\right) 
  + \sum_{i=1}^{k} c_1t^{\frac 1d}t^{i(\frac{d-1}{d})}
  \cdot \left(1 + \frac{F}{w}\log\frac{2(n/t^i+2F)}{F}\right)\label{eq:sum},
\end{align}
for some constants $c_0, c_1$.

Note that when $\frac{n}{t^i}\ge 2F$,
$\log\frac{2(n/t^i+2F)}{F}\le \log\frac{4n/t^i}{F}$
and so the summation in Inequality~\ref{eq:sum}  is upper bounded by
\[
  \sum_{i=1}^{k} c_1t^{\frac 1d}t^{i(\frac{d-1}{d})}
  \cdot \left(1 + \frac{F}{w}\log\frac{4n}{Ft^i}\right)
  =   c_1\sum_{i=1}^{k} (t^{\frac 1d}t^{i(\frac{d-1}{d})}) + t^{\frac 1d}t^{i(\frac{d-1}{d})}
  \cdot \left(\frac{F}{w}\log\frac{4n}{Ft^i}\right).
\]
Let $f(i)=t^{i(\frac{d-1}{d})}\log\frac{4n}{Ft^i}$, and we have
\[
	\frac{f(i+1)}{f(i)}=t^{\frac{d-1}{d}}\left(1+\frac{\log\frac{1}{t}}{\log\frac{4n}{Ft^i}}\right)
	= t^{\frac{d-1}{d}}\left(1+\frac{\log t}{\log\frac{Ft^i}{4n}}\right)
	\ge t^{\frac{d-1}{d}} \left(1+\frac{\log t}{\log\frac{w}{4}}\right)
	\ge w^{\delta(\frac{d-1}{d})}=\omega(1),
\]
where the first inequality follows from $\frac{n}{t^i}\ge\frac{F}{w}$
and the last inequality follows from $t=w^{\delta}$ for some constant $\delta$.
When $\frac{n}{t^i} < 2F$, the summation in Inequality~\ref{eq:sum}  is upper bounded by
\[
  \sum_{i=1}^{k} c_1t^{\frac 1d}t^{i(\frac{d-1}{d})}
  \cdot \left(1 + \frac{F}{w}\log\frac{2(2F+2F)}{F}\right)
  \le
  \sum_{i=1}^{k} c_1t^{\frac 1d}t^{i(\frac{d-1}{d})}
  \cdot \left(1 + \frac{F}{w}\log 8\right).
\]
In both cases, the summation is a geometric series and it is dominated by the last term, i.e.,
when $i$ achieves the maximum value.

Since $\frac{F}{w}\ge \frac{n}{\log^4n}$,
we observe that after constant levels of applying
Matou\v sek's scheme we will reach the base case because we 
stop as soon as the subproblem size goes below $\frac{F}{w}$. 
So we have $ \frac{F}{w^{1+\delta}} \le \frac{n}{t^k} \le \frac{F}{w}$ which 
means $t^k \le \frac{nw^{1+\delta}}{F}$ and $t^k \ge \frac{nw}{F}$.
Plugging in these inequalities in Equation~\ref{eq:sum}, and using that
$Q(n/t^k)= O(\frac{F}{w})$, we obtain the following bound  
\begin{align}
  Q(n) = O\left( \frac{n^{\frac{d-1}{d}}F^{\frac{1}{d}}}{w^{\frac{1}{d}-\delta}}\right).\label{eq:qm}
\end{align}

When $\frac{F}{w} < \frac{n}{\log^4n}$,
we use Chan's version (Theorem~\ref{thm:Chan}) to build the partition tree.
The observation here is that we can partition the point set into at least $\tau=\log^4n$ subsets.
Now the crossing number is at least $\log^2 n$ and thus the main term $O((b\tau)^{1-1/d})$ dominates.

We build a partition tree based on Chan's method (Theorem~\ref{thm:Chan}) with parameter $b=\Theta(1)$.
Again, we attach the corresponding frequency vectors to internal nodes.
The query time $Q(n)$ in this case is bounded by
\[
	Q(n) = O(n^{1-1/d})
	+ O\left( \sum_{i=0}^{\beta-1} bl(b^i) \cdot \left( 1+\frac{F}{w}\log\frac{2(\frac{n}{b^i}+2F)}{F} \right) \right),
\]
where the first term is the standard cost of traversing the partition tree in Chan's version
and the second term is the cost of examining all the disjoint cells fully contained in $q$,
and the parameters $\beta=\log_b\frac{n}{F/w}=\Omega(\log\log n)$ and $l(u)$ is the crossing number of $u$ cells.
Note that by a similar analysis as we did before, the summation is dominated by its last term, i.e.,
\begin{align*}
	\sum_{i=0}^{\beta-1} bl(b^i) \cdot \left( 1+\frac{F}{w}\log\frac{2(\frac{n}{b^i}+2F)}{F} \right)
	&= O\left( b \left(\frac{ nw}{b F}\right)^{\frac{d-1}{d}} \cdot \left( 1+\frac{F}{w}\log\frac{2(\frac{F}{w}+2F)}{F} \right) \right)\\
	&= O\left(\frac{n^{\frac{d-1}{d}}F^{\frac{1}{d}}}{w^{\frac{1}{d}}}\right).
\end{align*}
So $Q(n)=O \left(\frac{n^{\frac{d-1}{d}}F^{\frac{1}{d}}}{w^{\frac{1}{d}}}\right)$ in this case.
Note that the bound in Equation~\ref{eq:qm} is slightly larger and thus our query time is bounded
by Equation~\ref{eq:qm}.

\end{proof}

\section{Approximate Heavy Hitter Summary Queries}
\label{sec:hh}
In this section, we prove the following approximate heavy hitter summary query results for halfspace and dominance ranges.

\thmddd*
\thmdd*
\thmdomhh*

We start with the description of a base data structure for more general simplex ranges which will be used as a building block for our final AHHS solution. 

\subsection{Base Solution for AHHS Queries}
Our base data structure is for the following problem:
\probcoarse*

Note that this is a coarser problem compared to the original AHH query Problem~\ref{pr:ahhs}
since the error is fixed to be $\varepsilon|P|$ and is independent of the query range.
We prove the following theorem for this base problem:

\thmbase*

We will handle colors with different frequencies differently.
To simplify the exposition, we make the following definition:

\begin{definition}
	Let $P$ be a set of points with colors.
	Let $P_i$ be the set of points with color $i$.
	We say a color $i$ is \textbf{big} if $|P_i|\ge (\varepsilon |P|)^{\frac{2d}{d-1}}$;
	otherwise we say it is \textbf{small}.
	For a small color $i$,
    we say it is \textbf{$\delta$-small} if $\delta (\varepsilon |P|)^{\frac{2d}{d-1}} \le |P_i| \le 2\delta (\varepsilon |P|)^{\frac{2d}{d-1}}$.
\end{definition}

Note that a simple calculation shows that if $|P| > \varepsilon^{-\frac{2d}{d+1}}$, then there are no big
colors since the bound in the definition of a big color exceeds $|P|$.

\subparagraph{The data structure.} 
The idea is to count big colors \emph{exactly} and count small colors
\emph{approximately} through $\varepsilon$-approximations.
We collect the points with big colors in a set $P_\bbig$
and build a data structure $\Psi_\bbig$ on $P_\bbig$ using Theorem~\ref{thm:partition}. 
For each small color $i$, we first compute an $\varepsilon_i$-approximation $P'_i$ for $P_i$ 
where $\varepsilon_i = \varepsilon |P|/|P_i|$.
For a given value $\delta \le 1$, 
let $P_{\delta}$ be the set of points whose colors are $\delta$-small
and let 
$P'_\delta$ be the union of the approximations built on the $\delta$-small colors. 
We build $P'_\delta$ for
every $\delta = \frac{1}{2}, \frac{1}{4}, \frac{1}{8}, \cdots,\frac{1}{2}(\varepsilon|P|)^{-\frac{d+1}{d-1}}$. 
We can bound the number of colors that are $\delta$-small to be $F_\delta = 
\Theta(|P_\delta|/(\delta (\varepsilon |P|)^{\frac{2d}{d-1}}))$.
We also store $P'_\delta$ in a data structure $\Psi_\delta$ using Theorem~\ref{thm:partition}.

\subparagraph{The query algorithm.} The query is answered as follows. 
Given a query simplex range $q$,
we query the data structures $\Psi_\bbig$ as well as all the data structure
$\Psi_\delta$, for 
every $\delta = \frac{1}{2}, \frac{1}{4}, \frac{1}{8}, \cdots,\frac{1}{2}(\varepsilon|P|)^{-\frac{d+1}{d-1}}$. 
We examine each output frequency vector of colors and output
a color if its frequency is more than $\varepsilon|P|$.

The correctness is trivial since we are either counting the colors exactly or 
we are using an $\varepsilon_i$-approximation. In the latter case, the error
is at most $\varepsilon_i |P_i| = \varepsilon |P|$.
Since we can ignore colors of frequencies no more than $\varepsilon|P|$,
we only need to consider $\delta$-small colors for $\delta\ge\frac{1}{2}(\varepsilon|P|)^{-\frac{d+1}{d-1}}$.

\subparagraph{Space bound.}
Let $Y = |P_{\text{Big}} \cup \displaystyle\bigcup_{i \in small} P'_i|$.
The space usage of our data structure is clearly bounded by $O(Y)$.
We now show that $Y=O(X)$ which would prove our claim on the space bound. 

First we bound the size of the $\varepsilon_i$-approximation 
for each small color $i$. 
Consider a $\delta$-small color $i$ and the set $P_i$.
By construction, we are using an $\varepsilon_i = \varepsilon |P|/|P_i|$ approximation which has
size
\begin{align}
  O\left( \left(\frac{1}{\varepsilon_i}\right)^{\frac{2d}{d+1}} \right)
  &= O\left( \left( \frac{|P_i|}{\varepsilon |P|}\right)^{\frac{2d}{d+1}} \right)
  = O\left( \left( \frac{2\delta (\varepsilon |P|)^{\frac{2d}{d-1}}}{\varepsilon |P|} \right) ^{\frac{2d}{d+1}} \right) \nonumber\\
  &=  O\left( \left( \delta (\varepsilon |P|)^{\frac{d+1}{d-1}} \right) ^{\frac{2d}{d+1}}  \right)
  = O\left(\delta^{\frac{d-1}{d+1}} |P_i| \right) = O(|P_i|), \label{eq:esize}
\end{align}
where the second last equality follows from $\delta(\varepsilon|P|)^{\frac{2d}{d-1}} \le |P_i|$.
Thus, each $\varepsilon_i$-approximation has size at most $O(|P_i|)$ since $\delta\le1/2$, and so $Y\le|P|$. 

When $|P| > \varepsilon^{-\frac{2d}{d+1}}$, as we argued, there is no big color
and all the colors are small.
In this case, we can bound
\begin{align}
  Y \le \left|\bigcup_{i\in small} P'_i \right| 
  = \sum_{i} O\left( \left(\frac{1}{\varepsilon_i}\right)^{\frac{2d}{d+1}} \right)
  = O\left( \sum_{i} \left( \frac{|P_i|}{\varepsilon |P|}\right) ^{\frac{2d}{d+1}} \right) 
  = O\left( \left(\frac{1}{\varepsilon}\right)^{\frac{2d}{d+1}} \right),
\end{align}
where the last equality follows from the convexity of the function $x^{\frac{2d}{d+1}}$ and thus the sum is maximized when all the values
of $|P_i|$ are 0 and one of them equals $|P|$.
We have thus shown that $Y \le X= \min\{|P|, \varepsilon^{-\frac{2d}{d+1}}\}$
and so $O(X)$ is the space bound.

\subparagraph{Query time bound.}
It remains to analyze the query time. 
For $P_{Big}$,
note that the number of colors is no more than $F_{Big}=|P|/(\varepsilon|P|)^{\frac{2d}{d-1}}$.
By Theorem~\ref{thm:partition},
the data structure we built for $P_{Big}$ will have query time
\begin{align}
	O(|P_{Big}|^{(d-1)/d}+|P_{Big}|^{(d-1)/d}F_{Big}^{1/d}/w^{\alpha})
	&= O\left(|P|^{(d-1)/d}\right)+O\left(\frac{|P|^{\frac{d-1}{d}}|P|^{\frac{1}{d}}}{w^{\alpha}(\varepsilon|P|)^{\frac{2}{d-1}}}\right)\nonumber\\
	&= O\left(X^{(d-1)/d}\right)+O\left(\frac{|P|^{1-\frac{2}{d-1}}}{w^{\alpha}\varepsilon^{\frac{2}{d-1}}}\right)\nonumber,
\end{align}
for some constant $\alpha$,
where the first equality follows from $P_{Big}\subset P$
and the last equality follows from $|P|\le \varepsilon^{-\frac{2d}{d+1}}$ for big colors to exist
in which case $X=|P|$ by definition.
We now bound the query time for small colors.
By Equation~\ref{eq:esize}, the total size of the $\varepsilon$-approximations stored for the $\delta$-small
colors is $n_\delta = |P'_\delta| = O( \delta^{\frac{d-1}{d+1}} |P_\delta|)$.
By Theorem~\ref{thm:partition}, the query time of the data structure built on $P'_\delta$ is 
\begin{align}
O\left(\max\left\{ n_\delta^{(d-1)/d}, n_\delta^{(d-1)/d} \,F_\delta^{1/d}/w^\alpha\right\}\right) 
= O(n_\delta^{(d-1)/d}) + O (n_\delta^{(d-1)/d} \,F_\delta^{1/d}/w^\alpha).\label{eq:ddsum}
\end{align}
The total query time for small colors is the sum over all $\delta$ of the expression in Equation~\ref{eq:ddsum}.
We can bound each of the two terms separately. 

First observe that
\[
	\sum_{\delta} \left(\delta^{\frac{d-1}{d+1}} \left( \frac{|P_\delta|}{|P|} \right)\right)^{(d-1)/d} 
	\le \sum_{\delta} \delta^{\frac{(d-1)^2}{d(d+1)}}
	= O(1),
\]
since the summation is over geometrically decreasing values of $\delta$ and
$|P_\delta|\le |P|$ by definition.
When $X = \min\{|P|, \varepsilon^{-\frac{2d}{d+1}}\} = |P|$, this implies
\[
	\sum_\delta n_\delta^{\frac{d-1}{d}} = \sum_{\delta} (\delta^{\frac{d-1}{d+1}} |P_\delta|)^{(d-1)/d} \le |P|^{\frac{d-1}{d}} \le X^{\frac{d-1}{d}}.
\]

When $X = \min\{|P|, \varepsilon^{-\frac{2d}{d+1}}\} = \varepsilon^{-\frac{2d}{d+1}}$, by definition,
\begin{align*}
	\sum_\delta n_\delta^{\frac{d-1}{d}} = \sum_\delta \left(\sum_{i:\delta\textrm{-small}} \left(\frac{1}{\varepsilon_i}\right)^{\frac{2d}{d+1}}\right)^{\frac{d-1}{d}}
	&= \sum_\delta \left(\sum_{i:\delta\textrm{-small}} \left(\frac{|P_i|}{\varepsilon|P|}\right)^{\frac{2d}{d+1}}\right)^{\frac{d-1}{d}}
	\le \sum_\delta \left(\frac{|P_\delta|}{\varepsilon|P|}\right)^{\frac{2(d-1)}{d+1}}\\
	&\le \left(\frac{|P|}{\varepsilon|P|}\right)^{\frac{2(d-1)}{d+1}}
	= \left(\left(\frac{1}{\varepsilon}\right)^{\frac{2d}{d+1}}\right)^{\frac{d-1}{d}}=X^{\frac{d-1}{d}},
\end{align*}
where the first inequality follows from $\frac{2d}{d+1}\ge1$ for $d\ge1$ and the second inequality follows from $\frac{2(d-1)}{d+1}\ge 1$
when $d\ge 3$.

For the second term, since $F_{\delta} = O(|P_{\delta}|/(\delta(\varepsilon|P|)^{\frac{2d}{d-1}}))$,
\begin{align}
\sum_{\delta}O\left(\frac{n_\delta^{\frac{d-1}d} F_\delta^{\frac 1d}}{w^{\alpha}}\right) 
	&=\sum_{\delta} O\left(
			\frac{(\delta^{\frac{d-1}{d+1}} |P_\delta|)^{\frac{d-1}d} }{w^{\alpha}}\cdot
				\left( \frac{|P_\delta|}{\delta (\varepsilon |P|)^{\frac{2d}{d-1}}} \right)^{\frac 1d}
		\right)\nonumber
	=\sum_{\delta} O \left( \frac{\delta^{\frac{d-3}{d+1}} |P_\delta|}{w^{\alpha}(\varepsilon |P|)^{\frac{2}{d-1}}} \right)\\\nonumber
	&=O \left( \frac{ (\sum_{\delta} \delta^{\frac{d-3}{d+1}}) |P|}{w^{\alpha}(\varepsilon |P|)^{\frac{2}{d-1}}} \right)
	\le O \left( \frac{ |P|^{1-\frac{2}{d-1}}}{w^{\alpha}\varepsilon^{\frac{2}{d-1}}} \right),
\label{eq:qtime}
\end{align}
where the third equality follows from the observation that the sum of all $n_\delta$ is at most $|P|$
and the last equality follows from $d\ge 3$ and $\delta\le\frac{1}{2}$ forms a decreasing geometric series.

\subsection{3D AHHS Queries}
\label{sec:3dsol}
Now we proceed to show how to answer 3D heavy hitter queries
for halfspace and dominance ranges.
We show the following two results.

\thmddd*
\thmdomhh*

\subsubsection{The Data Structure for 3D AHHS Queries}
Let $h$ be the query halfspace (or dominance range) and let $k$ be the number of points in $h$ (i.e., $k = |h \cap P|$).
We consider the dual space. 
Let $H$ be the set of hyperplanes dual to $P$ and let $q$ be the dual of $h$.
Using Theorem~\ref{thm:approxrc}, we can find a $(1+\alpha)$-factor approximation, $k^*$, of $k$ in $O(\log n)$ time
and using linear space.
In our data structure, 
we will build a series of shallow cuttings, using parameters 
$\beta, \lambda$.
Depending on whether we are dealing with halfspace or dominance ranges, the values of these parameters will be
different. 
Nonetheless, since the value of $k^*$ is known, we can choose the correct level to query. 
To give more details, we consider the following cuttings. 

\begin{itemize}
    \item Base cutting: We build a \textit{base shallow cutting} for the $\varepsilon_0^{-1}$ level.
    \item Lower level cuttings: These are built for levels between  $\varepsilon_0^{-1}$ and
        $\varepsilon_0^{-1} w^{\beta\lambda}$, where $w=\Theta(\log n)$ is the word-size,
        and $\beta$ and $\lambda$ are parameters to be set later.
        To be specific, we build $\varepsilon_0^{-1}w^{\beta i}$-shallow cuttings for
        $i=1, \cdots, \lambda$. 
    
    \item Higher level cuttings: We build $\varepsilon_0^{-1}w^{\beta\lambda}2^i$-shallow cuttings for
        $i=1, \cdots, \log(n\varepsilon_0/w^{\beta\lambda})$. 
\end{itemize}

The data structure stores different auxiliary data structures depending on the shallow cutting. 
As a result, since the data structure knows $k^*$, it can choose the correct level to query.
For instance, if $k^* < \varepsilon_0^{-1}$, then the query is answered using the
auxiliary data structures built for the base cutting.
Similarly, when $\varepsilon_0^{-1} \le k^* \le \varepsilon_0^{-1} w^{\beta\lambda}$, the data structure uses
lower level cuttings, and otherwise it uses higher level cuttings. 

We describe each of them in turn. 
We start with the higher level cuttings.

\paragraph{Higher level cuttings.}
The data structure for these cuttings is simple.
For each cell $\cell$ in a higher level cutting,
we build a $c_1\varepsilon_0$-approximation $E_\cell$ for its conflict list $\sS_\cell$
and store it in the base structure with parameter $c_1\varepsilon_0$.

\subparagraph{The query algorithm.}
Recall that the query probes the higher level cuttings when $k^*$ exceeds $\varepsilon_0^{-1}w^{\beta\lambda}$.
In this case, it can find the cell $\cell \in \C_i$ such that $|\sS_\cell |= O(k^*)$. 
We query the base data structure built on $\sS_\cell$ and return the result as the AHHS.
The query time here is $O(\varepsilon_0^{-1})$ which follows directly from the base structure. 

\subparagraph{Correctness.}
Consider a color $c$ that has frequency $j \varepsilon_0 k$ among the hyperplanes that pass below $q$, for $j \ge 1$.
Since $E_\cell$ is a $c_1 \varepsilon_0$-approximation of $\sS_\cell$, 
it follows that $c$ appears $j\varepsilon_0 |E_\cell|\pm c_1\varepsilon_0 |E_\cell|$ times among the subset of $E_\cell$ that passes below $q$. 
If $c_1$ is chosen to be a small enough constant (e.g., $c_1 < 0.5$), 
$c$ appears frequent enough in $E$ that 
the base structure would report $c$ as well as its frequency with error at most $c_1\varepsilon_0|\sS_\cell|$.
This in turn allows us to approximate the frequency of $c$ within hyperplanes below $q$ up to an additive
error of $2c_1 \varepsilon_0 k$.

\subparagraph{Space usage.}
Here, we need to distinguish between halfspaces and dominance ranges. 
In the case of halfspaces, the size of $E_\cell$ is bounded by $O(\varepsilon_0^{-3/2})$ whereas for
dominance ranges this is bounded by $O(\varepsilon_0^{-1} \log^{4}\frac{1}{\varepsilon_0})$ by Theorem~\ref{thm:eapp}.  
Consequently, for halfspaces we have the following situation:
for a cell $\cell \in \C_i$ where $\C_i$ is a $\varepsilon_0^{-1}w^{\beta\lambda}2^i$-shallow cuttings,
according to Theorem~\ref{thm:base},
the space of the base structure built for $E_\cell$ is bounded by $O(\varepsilon_0^{-3/2})$ regardless and thus 
the total space consumption is bounded by 
\[
    \sum_{i=1}^{\infty} O\left(\varepsilon_0^{-3/2} \frac{n}{\varepsilon_0^{-1}w^{\beta\lambda}2^i}\right) = O(n)
\]
which follows by picking $\lambda$ such that $w^{\beta\lambda} = \varepsilon_0^{-1/2}$ for a small enough value $\beta$.

For dominance ranges, it is sufficient to pick $\lambda$ such that $\log^{4}\frac{1}{\varepsilon_0} \le w^{\beta\lambda} \le \varepsilon_0^{-1/2}$.
We will need to pick $\beta$ to be a small enough constant. 
However, this implies that $\lambda$ can also be chosen to be a constant for dominance cases. 
The total space consumption here is bounded by 
\[
    \sum_{i=1}^{\infty} O\left(\varepsilon_0^{-1}\log^{4}\frac{1}{\varepsilon_0}\cdot  \frac{n}{\varepsilon_0^{-1}\log^{4} \frac{1}{\varepsilon_0}\cdot 2^i}\right) = O(n).
\]

\paragraph{Lower level cuttings.}
For each cell $\cell$ in a lower level cutting, we store the list of \emph{frequent} colors that appear at least
$\frac{c_1\varepsilon_0}{w^\beta} |\sS_{\cell}|$ times in the conflict list $\sS_\cell$ of $\cell$. 
So for each cutting cell, we store $\frac{w^\beta}{c_1\varepsilon_0}$ frequent colors in the base structure in Theorem~\ref{thm:base},
and so the total space needed for all lower cuttings is
\[
    \sum_{i=1}^{\lambda} O\left( \frac{n}{w^{\beta i}/\varepsilon_0}\cdot\frac{w^\beta}{c_1\varepsilon_0} \right) = O(n).
\]

This also enables us to \emph{re-number} the colors: 
consider a $\cell$ in the lower level cuttings. 
We re-number the candidate colors in $\cell$ from 1 up to $w^{\beta}\varepsilon_0^{-1}$ and store them in a dictionary
such that given any (global) color $C$, we can fetch its index in our new re-numbering in $O(1)$ time. 
Re-numbering is crucial for using the base data structure as it returns a compact representation.
Note that the set of frequent colors is a super set of the actual heavy hitters.
Indeed, the base structures we built for lower levels can output a frequency vector of $\frac{w^{\beta}}{c_1\varepsilon_0}$ colors
(in a compact representation) and thus we cannot afford to check all these colors.

We use an auxiliary structure to generate a list of $O(\frac{1}{\varepsilon_0})$ \emph{candidate} colors for each query
and as we will show, we only need to check these candidate colors.
We call it the \emph{testing} structure.
This is also done via shallow cuttings:
We build shallow cuttings for level $\varepsilon_0^{-1}2^i$ for $i=1,2,\cdots,\log (w^{\beta\lambda})$.
We call these \emph{testing} cuttings.
For each cell $\cell$ in the testing cuttings,
we collect a list of colors that appear at least
$c_1\varepsilon_0 |\sS_{\cell}|$ times in the conflict list $\sS_\cell$ of $\cell$.
We call these colors \emph{candidate} colors
and clearly, for each cell $\cell$, we store at most $\frac{1}{c_1\varepsilon_0}$ candidate colors.
The total space, over all testing cuttings, needed to store the candidate colors is at most 
\[
    \sum_{i=1}^{\infty} O\left( \frac{n}{\varepsilon_0^{-1} 2^i}\cdot\frac{1}{c_1\varepsilon_0} \right) = O(n).
\]

\subparagraph{The query algorithm.}
Recall that the query probes the lower level cuttings when $k^*$ is between $\varepsilon_0^{-1}$ and $\varepsilon_0^{-1}w^{\beta\lambda}$.
In this case,  we find the smallest index $i$ such that $\varepsilon_0^{-1}w^{\beta i}$ exceeds $k^*$.
We find the cell $\cell$ in $\C_i$ that contains the query.
It thus follows that $|\sS_\cell |= O(w^\beta k)$. 
We query the base data structure implemented on $\sS_\cell$ with $q$.
However, since we have built the base structure with error parameter $\varepsilon_0/w^\beta$, 
the base data structure may return a list of $O(\varepsilon_0^{-1}w^\beta)$ colors, potentially in packed representation. 
We pick $\beta$ small enough such that 
the result of the base data structure fits in $\varepsilon_0^{-1}$ words. 
On the other hand, since $|\sS_\cell|=O(\varepsilon_0^{-1}w^{\beta\lambda})=O(\varepsilon_0^{3/2})$,
the query time of the base data structure is $O(\varepsilon_0^{-1})$ by Theorem~\ref{thm:base}.
In addition, we query the testing structure to obtain a list of $O(\varepsilon_0^{-1})$ candidate colors.
For each candidate color, we can use the stored dictionary to find its frequency in the
packed representation, if it exists, we report it using the extract operation.
This concludes the query algorithm. 

\subparagraph{Correctness.}
The base data structure finds all the colors that appear at least $\frac{c_1\varepsilon_0 |\sS_\cell|}{w^{\beta}} = 
\frac{O(c_1\varepsilon_0 k w^\beta)}{w^\beta} \le \varepsilon_0 k$ times below $q$, if $c_1$ is chosen to be a small
enough constant. 
Thus, any color that appears $\varepsilon_0 k$ times is reported correctly by the base data structure.

\subparagraph{Space analysis.}
By Theorem~\ref{thm:base}, each level in a lower level  shallow cutting $\C_i$ will consume
$O(n)$ space. 
There are $\lambda$ lower level shallow cuttings and thus the total space used by
lower level cuttings is $O(n \lambda)$. 

For halfspace ranges, recall that we needed to pick $\lambda$
such that $w^{\beta\lambda} = \varepsilon_0^{-1/2}$.
As $\beta$ is chosen to be a small value, it follows that for halfspaces we can choose
$\lambda = O(\log_w \frac{1}{\varepsilon_0})$,
leading to an $O(n\log_w \frac{1}{\varepsilon_0})$ space bound.

For dominance ranges, it is sufficient to pick $\lambda$ such that $w^{\beta\lambda}\ge\log^4\frac{1}{\varepsilon_0}$.
In this case, as we have shown before, by picking $\beta$ to be a small enough constant,
it suffices to pick $\lambda$ to be a large enough constant,
leading to an $O(n)$ space bound.

\paragraph{The base cutting.}
The base cutting, $\C_0$, is a $\varepsilon_0^{-1}$-shallow cutting.
We store the conflict list of the cells in $\C_0$ explicitly.
By Lemma~\ref{lem:sc}, in this case, we find $\C_0$ and the cell $\cell\in\C_0$ that contains $q$ in time $O(\log n)$.
Here, we can explicitly access $\sS_\cell$ and answer the AHHS query directly,
in time $O(\log n + \varepsilon_0^{-1})$ with no error.

\subsubsection{Putting It Together}
Over all three different shallow cutting levels, the total space complexity was
$O(n\lambda)$.
For dominance ranges, we saw that we can afford to pick $\lambda$ to be a constant.
This yields a data structure with $O(n)$ space.
For halfspaces, we picked
$\lambda = O(\log_w \frac{1}{\varepsilon_0})$
and thus we get a data structure with 
$O(n\log_w \frac{1}{\varepsilon_0})$ space. 

\subsection{2D AHHS Queries}
In the plane, we can actually improve the space complexity even further.
However, this requires modifying the base structure.
The main idea here is that the size of $\varepsilon$-approximation in 2D is small enough that we can 
try more aggressive approximations. 

\subsubsection{Base Solution for 2D AHHS Queries}
This subsection is devoted to the proof of the specialized base structure for 2D. 

\begin{theorem}\label{thm:base2d}
    Consider a colored point set $P$ in $\R^2$ and let $\varepsilon$  be parameter.
    Let $B=\varepsilon|P|$.
    We can store $P$ in a data structure of linear size such that given a set $q_c$ of $F=\frac{1}{\varepsilon B^{1/3}}$ colors,
    one can estimate the frequency of the colors in $q_c$ up to additive error of $\varepsilon |P|$ in $O(F)$ time.
\end{theorem}

\subparagraph{The main idea and a summary.}
The main observation is that in 2D, given a point set $P$ and a value $\varepsilon$, there exists
$\varepsilon$-approximations of size $O(\varepsilon^{-4/3})$ and thus if we store them in a data structure
for halfspace range counting queries, the query time will be reduced to $O(\varepsilon^{-2/3})$ i.e., smaller
than $\varepsilon^{-1}$.
Intuitively, this means that the difficult case of the problem is when there a lot of different colors with low
frequency. However, proving the base theorem is still quite non-trivial since we are aiming for very fast query times.
Our main observation is that when there are few points, the size of the $\varepsilon$-approximations  are small
enough that we can afford to spend slightly higher space and lower the query time in return.
By picking the parameters carefully, we arrive at the claimed theorem. 

\subparagraph{Dealing with ``frequent'' colors.}
We define a color to be frequent if it appears at least $\varepsilon B|P|$ times.
Clearly, the number of frequent colors is at most $|P|/(\varepsilon B|P|)= (\varepsilon B)^{-1}$.
Consider a frequent color $j$ and assume it appears $n_j$ times in the point set $P$.
We store an $\varepsilon_j$-approximation $\eps_j$ for this color where $\varepsilon_j = \frac{\varepsilon |P|}{n_j}$.
The size of $\eps_j$ is $\min\left\{ n_j,O\left( \frac{n_j}{\varepsilon |P|} \right)^{4/3} \right\}$ and we store
$\eps_j$ in a data structure in Theorem~\ref{thm:simplexrc}.

Now consider a query and consider dealing with the frequent colors in $q_c$.
For every frequent color $j \in q_c$, we simply count the number of points from $\eps_j$
and use it to get an estimate with the correct additive error.
The query time is $O(\sqrt{|\eps_j|}) = O\left( \frac{n_j}{\varepsilon |P|} \right)^{2/3}$.
The total query time spent on all the frequent colors is thus
\begin{align*}
    \sum_{j \in q_c \mbox{, $j$ is frequent}} O\left( \frac{n_j}{\varepsilon |P|} \right)^{2/3} \le \frac{1}{\varepsilon B} O\left( \frac{\varepsilon|P|  B}{\varepsilon |P|} \right)^{2/3} = O\left( \frac{1}{\varepsilon B^{1/3} }\right)
\end{align*}
where the inequality follows from the observation that $\sum {n_j} \le |P|$ and 
since the exponent $2/3$ is less than one, the function $x^{2/3}$ is a concave function
and thus the expression on the left is maximized when all the $n_j$'s are equal.
Since there are at most $(\varepsilon B)^{-1}$ colors the maximum is when 
each $n_j =|P|/ (\varepsilon B)^{-1} = \varepsilon B |P|$.
As a result, we can deal with the frequent colors within our claimed query time. 
The space is also clearly linear. 

\subparagraph{Dealing with infrequent colors.}
This is the more tricky part of the data structure. 
Consider an infrequent color $j$ that appears $n_j$ times. 
As before, we build an $\varepsilon_j$-approximation $\eps_j$ for this color where $\varepsilon_j = \frac{\varepsilon |P|}{n_j}$
and the size of $\eps_j$ is $O\left( \frac{n_j}{\varepsilon |P|} \right)^{4/3}$.
Let $X$ be the collection of $\eps_j$ for all the infrequent colors $j$.
We can bound the size of $X$ asymptotically  by
\begin{align*}
    \sum_{\mbox{$j$ is infrequent}} \left( \frac{n_j}{\varepsilon |P|} \right)^{4/3} \le \frac{1}{\varepsilon B} \left( \frac{\varepsilon B |P|}{\varepsilon |P|} \right)^{4/3} = \frac{B^{1/3}}{\varepsilon}
\end{align*}
where the inequality follows from Jensen's inequality, i.e.,  that the function $x^{4/3}$ is a convex function that thus the expression on the left is maximized when 
each $n_j$ is either maximized or minimized;
however, as each infrequent color appears at most $\varepsilon B |P|$ times, the maximum is achieved when $\frac{1}{\varepsilon B}$ of the $n_j$'s are set to
$\varepsilon B |P|$ and the rest are set to 0.

We then transform $X$ into the dual space using point-line duality.
We build a $B^{-2/3}$-cutting $Z$ in the dual space. 
By Lemma~\ref{lem:cutting}, $Z$ has $O(B^{4/3})$ triangles that cover the entire plane where 
each triangle is intersected by $O(X/B^{2/3}) = O( 1/(\varepsilon B^{1/3}))$ lines. 
For each triangle, we store the number of lines of each color that passes below it ($F =  O( 1/(\varepsilon B^{1/3}))$ frequencies), 
as well as the set of lines intersecting it ($ O( 1/(\varepsilon B^{1/3}))$ lines).
In total we will asymptotically use 
\[
    \frac{1}{\varepsilon B^{1/3}}\cdot B^{4/3} = \frac{B}{\varepsilon} \le |P|
\]
space.
To bound the query time, consider a query point $q$ which is dual to the query halfspace.
Now observe that after locating the triangle in the cutting that contains $q$, we simply need to 
look at the stored frequencies in the triangle, as well as the set of lines that intersect the triangle.
In total this will take $O(1/(\varepsilon B^{1/3}))$ time. 

\subsubsection{A Data Structure for 2D AHHS Queries}
In our data structure for 2D queries, we use most of the ingredients that we developed for 3D.
In particular, we can use $O(n)$ space such that for a given halfplane $h$, we can find a 
list of $O(\varepsilon_0^{-1})$ candidate colors in $O(\log n + \varepsilon_0^{-1})$ time. 
Similarly, we also build a level 0 shallow cutting $\C_0$ that allows us answer the query
when $h$ contains at most $\varepsilon_0^{-1}$ points. 
However, we only build the first lower level shallow cutting $\C_1$ from Subsection~\ref{sec:3dsol}.
Consequently, if $h$ contains up to $\varepsilon_0^{-1} w^\beta$ points, we can answer the query 
in optimal running time. 
Handling the remaining queries is where we deviate from the 3D data structure. 

We now build a different hierarchy of shallow cuttings.
Let $g_0 = \varepsilon_0^{-1}w^\beta$. 
We build a series of shallow cuttings using parameters $g_0, g_1, \cdots$ that will be
determined shortly. 
In particular, we build a $g_i$-shallow cutting $\C'_i$ until $g_i$ exceeds $\varepsilon_0^{-3}$;
we then switch back to the 3D solution and as argued there, the total space of this part of
the structure will be $O(n)$.
For each cell $\cell$ in $\C'_i$, we store the conflict list of $\cell$ in the base structure of Theorem~\ref{thm:base2d},
with the following parameters:
we have $|P| = \Theta(g_i)$, we want to
set $\varepsilon$ such that the query time equals $O(1/\varepsilon_0)$ and thus 
we must satisfy 
\begin{align}
  \frac{1}{\varepsilon  B^{1/3}} \le \frac{1}{\varepsilon_0} \Leftrightarrow {\varepsilon  B^{1/3}} \ge \varepsilon_0 \Leftrightarrow {\varepsilon  \left( \varepsilon |P| \right)^{1/3}} \ge {\varepsilon_0} \Leftarrow {\varepsilon \left( \varepsilon g_i \right)^{1/3}} \ge {\varepsilon_0}\label{eq:qtime}
\end{align}
Thus, setting our parameters in way to satisfy Inequality~\ref{eq:qtime} makes sure that the query time 
of the base structure we are using is $O(\varepsilon_0^{-1})$.

To guarantee the approximation factor, observe that we know that the query is outside the shallow cutting $\C'_{i-1}$ which implies
there are at least $g_{i-1}$ points in the query. 
Consequently, if we can ensure the following, then our approximation factor is also as desired:
\begin{align}
  \varepsilon|P| \le \varepsilon_0 g_{i-1} \Leftarrow \varepsilon \Theta(g_i) \le \varepsilon_0 g_{i-1}. \label{eq:approx}
\end{align}
Inequalities~\ref{eq:qtime} and~\ref{eq:approx} give us a recursion for $g_i$.
To simplify the exposition, we can simply rescale the constant $\varepsilon_0$ 
to get rid of the constant in the $\Theta(\cdot)$ notation in Inequality~\ref{eq:approx} and 
turn the inequality into an equality. 
This yields $\varepsilon = \frac{\varepsilon_0 g_{i-1}}{g_i}$ which in turn yields 
\begin{align}
  g_i  \le g_{i-1}  \left( \varepsilon_0 g_{i-1} \right)^{\frac 13}\quad \mbox{ and recall that we have } g_0 = \frac{w^\beta}{\varepsilon_0}.\label{eq:eps}
\end{align}
Now a simple induction yields that
\begin{align}
  g_i = \frac{1}{\varepsilon_0} \cdot \left( w^\beta \right)^{(4/3)^{i}}\label{eq:sol}.
\end{align}
Recall that as soon as $g_i$ reaches $\varepsilon_0^{-3}$, we switch to the 3D solution and from this
point on, the remaining levels will consume $O(n)$ space. 
As a result, the space complexity of the data structure is $O(n j)$ where $j$ is the smallest index
such that $g_j \ge \varepsilon_0^{-3}$.
From Equation~\ref{eq:sol}, it is clear that $j = O(\log\log_w \varepsilon_0^{-1})$ and thus
the space bound is $O(n\log\log_w \varepsilon_0^{-1})$.

\end{appendices}